\documentclass[11pt,amssymb,amsfont,a4paper]{article}
\usepackage{latexsym,amssymb,amsmath,amsthm,color}
\usepackage{geometry}
\usepackage{url}
\usepackage{graphicx}
\usepackage[utf8]{inputenc}
\usepackage[numbers]{natbib}
\usepackage{authblk}
\usepackage{tikz}

\theoremstyle{definition}
\newtheorem{definition}{Definition}

\theoremstyle{plain}
\newtheorem{theorem}{Theorem}
\newtheorem{proposition}{Proposition}
\newtheorem{lemma}{Lemma}
\newtheorem{remark}{Remark}
\newtheorem{corollary}{Corollary}

\title{Theory of supports for linear codes endowed with the sum-rank metric}
\author[1]{Umberto Mart{\'i}nez-Pe\~{n}as \thanks{umberto@ece.utoronto.ca}}
\affil[1]{Dept.\ of Electrical \& Computer Engineering,
University of Toronto, Canada}

\date{}

\begin{document}

\maketitle

\begin{abstract}
The sum-rank metric naturally extends both the Hamming and rank metrics in coding theory over fields. It measures the error-correcting capability of codes in multishot matrix-multiplicative channels (e.g. linear network coding or the discrete memoryless channel on fields). Although this metric has already shown to be of interest in several applications, not much is known about it. In this work, sum-rank supports for codewords and linear codes are introduced and studied, with emphasis on duality. The lattice structure of sum-rank supports is given; characterizations of the ambient spaces (\textit{support spaces}) they define are obtained; the classical operations of restriction and shortening are extended to the sum-rank metric; and estimates (bounds and equalities) on the parameters of such restricted and shortened codes are found. Three main applications are given: 1) Generalized sum-rank weights are introduced, together with their basic properties and bounds; 2) It is shown that duals, shortened and restricted codes of maximum sum-rank distance (MSRD) codes are in turn MSRD; 3) Degenerateness and effective lengths of sum-rank codes are introduced and characterized. In an appendix, skew supports are introduced, defined by skew polynomials, and their connection to sum-rank supports is given.

\textbf{Keywords:} Generalized sum-rank weights, Hamming metric, MSRD codes, multishot matrix-multiplicative channel, rank metric, sum-rank metric, sum-rank support, wire-tap channel.

\textbf{MSC:} 94A60; 94B05; 94C99.
\end{abstract}

\section{Introduction} \label{sec intro}

The Hamming metric has played a major role in the theory of error-correcting codes since its beginnings \cite{hamming}. Its $ q $-analog, the rank metric \cite{delsartebilinear, gabidulin, roth}, has brought considerable attention due to its natural connection with universal error correction in linear network coding \cite{errors-network}, since a linearly coded network is a matrix-multiplicative channel.

A common extension (over fields) of both the Hamming and rank metrics, called \textit{sum-rank metric}, was implicitly considered in the space-time coding literature \cite{algebraic-space-time, space-time-kumar}, and formally defined for multishot network coding in \cite{multishot}. Applications of codes endowed with the sum-rank metric (\textit{sum-rank codes}) include linear network coding \cite{secure-multishot, multishot}, space-time coding \cite{algebraic-space-time, space-time-kumar} and distributed storage \cite{universal-lrc}. 

Several constructions of sum-rank codes appeared in \cite{napp-sidorenko, multishot2009, multishot, wachter, wachter-convolutional}. \textit{Convolutional} codes with maximum sum-rank distance appeared in \cite{mahmood-convolutional, mrd-convolutional}. \textit{Linearized Reed-Solomon codes} \cite{linearizedRS} constitute the first and only known family of maximum sum-rank distance (MSRD) \textit{block} codes whose field sizes are not exponential in the code length. Their application in reliable and secure multishot network coding, plus a sum-rank decoding algorithm with quadratic complexity, were given in \cite{secure-multishot}. Maximum rank distance (MRD) block codes, such as Gabidulin codes \cite{gabidulin, roth}, are also MSRD codes but always require field sizes that are exponential in the code length, hence being disadvantageous (see also \cite[Subsec. 4.2]{linearizedRS} or \cite[Table I]{secure-multishot}). Finally, \textit{sum-rank alternant codes} were introduced in \cite{universal-lrc} to obtain sum-rank codes over smaller fields at the expense of not being MSRD.

As is well-known, Gabidulin codes \cite{gabidulin, roth} cover all possible parameters for which MRD codes exist. However, the situation is far from clear for MSRD codes. Indeed, knowing all parameters for which MSRD codes exist solves as a particular case the MDS conjecture, see \cite[Sec. VI]{secure-multishot}. By the connection made in \cite{universal-lrc}, it would also further the knowledge on the existence of \textit{Partial-MDS} codes \cite{blaum-RAID, gopalan-MR}. 

Therefore, sum-rank codes are of both practical and theoretical interest. However, not many constructions exist and very little is known about their structure. In this work, we introduce and study \textit{sum-rank supports} for codewords and linear codes, with an emphasis on duality. We introduce and characterize \textit{sum-rank support spaces}, which give the ambient spaces for linear sum-rank codes. We then extend the classical operations of \textit{restriction} and \textit{shortening} to linear sum-rank codes, establish their duality and give estimates on the sum-rank parameters of the codes that they provide. We give three applications: 1) We introduce generalized sum-rank weights, give their basic properties and applications, and establish a hierarchy of bounds on them, extending a previous connection between bounds on generalized Hamming and rank weights \cite{similarities}; 2) We connect the MSRD property with sum-rank supports, and prove that duals, shortenings and restrictions of MSRD codes give again MSRD codes; 3) We introduce and characterize degenerateness and effective length of linear sum-rank codes, showing in particular that MSRD codes are never sum-rank degenerate.

All these sum-rank coding-theoretic results particularize to well-known results for Hamming-metric and rank-metric codes. We will show how throughout the paper. As a final remark, it is worth noting that sum-rank codes fall outside the framework of $ q $-analogues in coding theory. Indeed, the lattice of sum-rank supports extends simultaneously the lattice of subsets of a finite set, and that of subspaces of a finite-dimensional vector space (see Subsection \ref{subsec lattice unions spaces}). Interestingly, the concept of sets of (arithmetic) roots of skew polynomials, as introduced in \cite{lam, lam-leroy}, gives a uniformizing framework for all these lattices. This was used in \cite{linearizedRS} to prove that linearized Reed-Solomon codes are MSRD by relating them with skew Reed-Solomon codes \cite{skew-evaluation1}. It was also used in \cite{secure-multishot} to give a uniform arithmetic description of the sum-rank version of the Welch-Berlekamp decoding algorithm. All the concepts and results in this paper could be stated in terms of roots of skew polynomials (e.g. \textit{skew supports} and \textit{generalized skew weights}) in the particular case where the base subfields are centralizers of an endomorphism and derivation of the extension field (more general cases would require the use of multivariate skew polynomials \cite{lin-multivariateskew, multivariateskew}). See Appendix \ref{app} for more details.

The organization of the remainder of the paper is as follows. In Section \ref{sec the sum-rank metric and support}, we recall the definition of sum-rank metric from \cite{multishot} (Subsection \ref{subsec sum-rank supports and weights}), and we introduce its support lattice (Subsection \ref{subsec lattice unions spaces}). In Section \ref{sec sum-rank support spaces}, we introduce and characterize sum-rank support spaces (Subsection \ref{subsec equivalent defs sum-rank support spaces}), we study their lattice structure (Subsection \ref{subsec properties sum-rank support spaces}), and we define and study restricted and shortened sum-rank codes (Subsection \ref{subsec restricted and shortened}). In Section \ref{sec some applications}, we introduce and study generalized sum-rank weights (Subsection \ref{subsec gen sum-rank weights}), MSRD codes (Subsection \ref{subsec MSRD codes}), and sum-rank degenerateness and effective length (Subsection \ref{subsec sum-rank degenerate}). In Appendix \ref{app}, we introduce skew supports and support spaces, and provide their connection to sum-rank supports and support spaces.

\section{The sum-rank metric and its support lattice} \label{sec the sum-rank metric and support}

In this section, we recall the concept of sum-rank metric from \cite{multishot} (Subsection \ref{subsec sum-rank supports and weights}) and introduce its support lattice (Subsection \ref{subsec lattice unions spaces}).

Fix positive integers $ \ell $, $ m_1 $, $ m_2 $, $ \ldots $, $ m_\ell $ and $ n_1 $, $ n_2 $, $ \ldots $, $ n_\ell $, and fix fields $ K_1 $, $ K_2 $, $ \ldots $, $ K_\ell $. In principle, codes endowed with the sum-rank metric are simply subsets of 
$$ K_1^{m_1 \times n_1} \times K_2^{m_2 \times n_2} \times \cdots \times K_\ell^{m_\ell \times n_\ell}. $$
In this way, a word is just a list of matrices. The \textit{sum-rank distance} between two such lists, $ (C_1, C_2, \ldots, C_\ell) $ and $ (D_1, D_2, \ldots, D_\ell) $, may be defined as $ \sum_{i=1}^\ell {\rm Rk}_{K_i}(C_i - D_i) $. 

In this work, we will assume that all fields $ K_i $ are subfields of a larger field $ \mathbb{F} $ with $ m_i = \dim_{K_i}(\mathbb{F}) $, for $ i = 1,2, \ldots, \ell $, and we will represent each word as a vector in $ \mathbb{F}^n $, where $ n = n_1 + n_2 + \cdots + n_\ell $. However, the definitions of supports and support lattice in this section remain valid in the general case, even if the fields have different characteristic.

For a field $ K $, we will denote by $ K^{m \times n} $ the set of $ m \times n $ matrices with coefficients in $ K $, and we denote $ K^n = K^{1 \times n} $. The field over which we consider linearity, ranks and dimensions will be assumed from the context. We will also denote by $ {\rm Row}(M) \subseteq K^n $ and $ {\rm Col}(M) \subseteq K^m $ the row and column spaces of a matrix $ M \in K^{m \times n} $.

We will implicitly consider ``erased'' matrices, which may be denoted by $ * $. We will define $ K^{m \times 0} = \{ * \} $ and $ K^{0 \times n} = \{ * \} $. Operations with matrices are assumed to be trivially extended to $ * $. For instance, $ {\rm Rk}(*) = 0 $ or $ A* = * $ for $ A \in K^{n \times m} $. However, if $ A \in K^{m_1 \times n_1} $, $ * \in K^{0 \times n_2} $, $ B \in K^{m_3 \times n_3} $, we define $ {\rm diag}(A,*,B) \in K^{(m_1 + m_3) \times (n_1 + n_2 + n_3)} $ as putting $ n_2 $ zero columns between the first $ n_1 $ and the last $ n_3 $ columns in $ {\rm diag}(A,B) $.

\subsection{The sum-rank metric} \label{subsec sum-rank supports and weights}

Fix an ordered basis $ \mathcal{A}_i = \{ \alpha_1^{(i)}, \alpha_2^{(i)}, \ldots, \alpha_{m_i}^{(i)} \} $ of $ \mathbb{F} $ over $ K_i $, where $ m_i = \dim_{K_i}(\mathbb{F}) $, for $ i = 1,2, \ldots, \ell $. For any non-negative integer $ r $, we denote by $ M_{\mathcal{A}_i} : \mathbb{F}^r \longrightarrow K_i^{m_i \times r} $ the corresponding matrix representation map, given by 
\begin{equation}
M_{\mathcal{A}_i} \left( \sum_{j=1}^{m_i} \alpha_j^{(i)} \mathbf{c}_j \right) = \left( \begin{array}{c}
\mathbf{c}_1 \\
\mathbf{c}_2 \\
\vdots \\
\mathbf{c}_{m_i}
\end{array} \right) = \left( \begin{array}{cccc}
c_{11} & c_{12} & \ldots & c_{1 r} \\
c_{21} & c_{22} & \ldots & c_{2 r} \\
\vdots & \vdots & \ddots & \vdots \\
c_{m_i 1} & c_{m_i 2} & \ldots & c_{m_i r} \\
\end{array} \right) \in K_i^{m_i \times r},
\label{eq def matrix representation map}
\end{equation}
where $ \mathbf{c}_j = (c_{j1}, c_{j2}, \ldots, c_{jr}) \in K_i^r $, for $ j = 1,2, \ldots, m_i $, and for $ i = 1,2, \ldots, \ell $. We may now define the sum-rank metric in $ \mathbb{F}^n $, which was introduced in \cite[Subsec. III-D]{multishot} under the name \textit{extended distance} or \textit{multishot rank metric}.

\begin{definition} [\textbf{Sum-rank metric \cite{multishot}}] \label{def sum-rank metric}
Let $ \mathbf{c} = (\mathbf{c}^{(1)}, \mathbf{c}^{(2)}, \ldots, $ $ \mathbf{c}^{(\ell)}) \in \mathbb{F}^n $, where $ \mathbf{c}^{(i)} \in \mathbb{F}^{n_i} $, for $ i = 1,2, \ldots, \ell $. We define the sum-rank weight of $ \mathbf{c} \in \mathbb{F}^n $ as
$$ {\rm wt}_{SR}(\mathbf{c}) = \sum_{i=1}^\ell {\rm Rk}_{K_i}(M_{\mathcal{A}_i}(\mathbf{c}^{(i)})). $$ 
We then define the sum-rank metric $ {\rm d}_{SR} : (\mathbb{F}^n)^2 \longrightarrow \mathbb{N} $ as $ {\rm d}_{SR}(\mathbf{c}, \mathbf{d}) = {\rm wt}_{SR}(\mathbf{c} - \mathbf{d}) $, for all $ \mathbf{c}, \mathbf{d} \in \mathbb{F}^n $. For a linear code $ \mathcal{C} \subseteq \mathbb{F}^n $, we define its minimum sum-rank distance as $ {\rm d}_{SR} (\mathcal{C}) = \min \{ {\rm wt}_{SR}(\mathbf{c}) \mid \mathbf{c} \in \mathcal{C} \setminus \{ \mathbf{0} \} \} $.
\end{definition}

Observe that $ {\rm Rk}_{K_i}(M_{\mathcal{A}_i}(\mathbf{c}^{(i)})) $ does not depend on $ \mathcal{A}_i $, for $ i = 1,2, \ldots, \ell $. The parameter $ \ell $ represents the number of shots in a matrix-multiplicative channel \cite{multishot}. Observe that the rank metric \cite{delsartebilinear, gabidulin, roth} is recovered when $ \ell = 1 $, that is, the rank metric is a singleshot sum-rank metric. The Hamming metric \cite{hamming} is recovered when $ \mathbf{n} = (n_1, n_2, \ldots, n_\ell) = (1,1, \ldots, 1) = \mathbf{1} $, the vector whose entries are all $ 1 $, since the alphabet is a field. That is, the Hamming metric over fields is a multishot ``$ (1 \times 1) $-rank'' metric. Therefore, the sum-rank metric is a natural extension of both the rank and Hamming metrics. Throughout the paper, we will show how our results extend known results for both metrics at the same time. We will also use the notation $ {\rm wt}_R = {\rm wt}_{SR} $ and $ {\rm d}_R = {\rm d}_{SR} $ if $ \ell = 1 $, and $ {\rm wt}_H = {\rm wt}_{SR} $ and $ {\rm d}_H = {\rm d}_{SR} $ if $ \mathbf{n} = \mathbf{1} $.

\subsection{The lattice of lists of vector spaces} \label{subsec lattice unions spaces}

In this subsection, we introduce a lattice that is a hybrid between the lattice of subsets of a finite set, and that of subspaces of a finite-dimensional vector space. We will conclude by showing that this lattice gives the natural supports for defining sum-rank weights of vectors. From now on, we will denote $ \mathbf{K} = (K_1, K_2, \ldots, K_\ell) $ and $ \mathbf{n} = (n_1, n_2, \ldots, n_\ell ) $.

\begin{definition} \label{def cartesian product lattice}
We define the cartesian product lattice
$$ \mathcal{P}(\mathbf{K}^\mathbf{n}) = \mathcal{P}(K_1^{n_1}) \times \mathcal{P}(K_2^{n_2}) \times \cdots \times \mathcal{P}(K_\ell^{n_\ell}), $$
where $ \mathcal{P}(K_i^{n_i}) $ is the lattice of $ K_i $-linear vector subspaces of $ K_i^{n_i} $. 
\end{definition}

This set forms a lattice as partially ordered set with inclusions as in the following definition. We also include other basic operations that will be useful later.

\begin{definition} \label{def basic operations lattice}
Fix $ \boldsymbol{\mathcal{L}} = (\mathcal{L}_1, \mathcal{L}_2, \ldots, \mathcal{L}_\ell) $ and $ \boldsymbol{\mathcal{L}}^\prime = (\mathcal{L}_1^\prime, \mathcal{L}_2^\prime, \ldots, \mathcal{L}_\ell^\prime) $ in $ \mathcal{P}(\mathbf{K}^\mathbf{n}) $. We define:
\begin{enumerate}
\item
Ranks: $ {\rm Rk}(\boldsymbol{\mathcal{L}}) = \sum_{i=1}^\ell \dim_{K_i}(\mathcal{L}_i) $.
\item
Inclusions: $ \boldsymbol{\mathcal{L}} \subseteq \boldsymbol{\mathcal{L}}^\prime $ if, and only if, $ \mathcal{L}_i \subseteq \mathcal{L}_i^\prime $, for all $ i = 1,2, \ldots, \ell $.
\item
Sums: $ \boldsymbol{\mathcal{L}} + \boldsymbol{\mathcal{L}}^\prime = (\mathcal{L}_1 + \mathcal{L}_1^\prime, \mathcal{L}_2 + \mathcal{L}_2^\prime, \ldots, \mathcal{L}_\ell + \mathcal{L}_\ell^\prime) $.
\item
Intersections: $ \boldsymbol{\mathcal{L}} \cap \boldsymbol{\mathcal{L}}^\prime = (\mathcal{L}_1 \cap \mathcal{L}_1^\prime, \mathcal{L}_2 \cap \mathcal{L}_2^\prime, \ldots, \mathcal{L}_\ell \cap \mathcal{L}_\ell^\prime) $.
\item
Duals: $ \boldsymbol{\mathcal{L}}^\perp = (\mathcal{L}_1^\perp, \mathcal{L}_2^\perp, \ldots, \mathcal{L}_\ell^\perp) $.
\item
Total and zero spaces: $ \boldsymbol{\mathcal{T}} = (K_1^{n_1}, K_2^{n_2}, \ldots, K_\ell^{n_\ell}) $ and $ \mathbf{0} = (\{ \mathbf{0} \}, \{ \mathbf{0} \}, \ldots, \{ \mathbf{0} \}) $, respectively.
\item
Direct sums: $ \boldsymbol{\mathcal{L}}^{\prime \prime} = \boldsymbol{\mathcal{L}} \oplus \boldsymbol{\mathcal{L}}^\prime $ if $ \boldsymbol{\mathcal{L}}^{\prime \prime} = \boldsymbol{\mathcal{L}} + \boldsymbol{\mathcal{L}}^\prime $ and $ \mathbf{0} = \boldsymbol{\mathcal{L}} \cap \boldsymbol{\mathcal{L}}^\prime $.
\item
Complementaries: $ \boldsymbol{\mathcal{L}}^\prime $ is a complementary of $ \boldsymbol{\mathcal{L}}^{\prime \prime} \in \mathcal{P}(\mathbf{K}^\mathbf{n}) $ in $ \boldsymbol{\mathcal{L}} $ if $ \boldsymbol{\mathcal{L}} = \boldsymbol{\mathcal{L}}^\prime \oplus \boldsymbol{\mathcal{L}}^{\prime \prime} $.
\end{enumerate}
\end{definition}

Observe that the sums and intersections defined above are indeed the smallest and largest elements of the lattice containing or contained in $ \boldsymbol{\mathcal{L}} $ and $ \boldsymbol{\mathcal{L}}^\prime $, respectively. Thus such sums and intersections are the natural lattice operations corresponding to the given partial order.

The reader will notice that, if $ K = K_1 = K_2 = \ldots = K_\ell $, then the previous lattice is essentially that of subspaces $ \mathcal{V} $ of $ K^n $ that can be decomposed as $ \mathcal{V} = \mathcal{V}_1 \times \mathcal{V}_2 \times \cdots \times \mathcal{V}_\ell $, where $ \mathcal{V}_i \subseteq K^{n_i} $, for $ i = 1,2, \ldots, \ell $. For different fields (even with different characteristic), we could represent such lists as cartesian products of certain $ \mathbb{Z} $-linear modules, but we would still need to keep track of dimensions, linearity and duality over the field $ K_i $ in the $ i $th position. Thus we prefer to keep the list-type representation of this lattice. 

The lattice $ \mathcal{P}(\mathbf{K}^\mathbf{n}) $ naturally recovers both the lattice $ \mathcal{P}([\ell]) $ of subsets of $ [\ell] = \{ 1,2, \ldots, \ell \} $ (setting $ \mathbf{n} = \mathbf{1} $) and the lattice $ \mathcal{P}(\mathcal{V}) $ of $ K $-linear subspaces of a vector space $ \mathcal{V} $ over a field $ K $ (setting $ \ell = 1 $), both lattices considered with conventional inclusions. The previous list of definitions then specializes to the standard definitions in these lattices, being ranks called \textit{sizes} when $ \mathbf{n} = \mathbf{1} $ and \textit{dimensions} when $ \ell = 1 $. To see this in the case $ \mathbf{n} = \mathbf{1} $, note that elements in $ \mathcal{P}(\mathbf{K}^\mathbf{1}) $ are lists in $ \prod_{i=1}^\ell \{ K_i, \{ 0 \} \} $. It is straightforward to translate such lists to binary strings, with elements in $ \{ 1,0 \} $, by taking dimensions over each of the fields $ K_i $. We may then translate these to subsets of $ I \subset [\ell] $ by means of the indicator function. Observe that duals and complementaries coincide and are uniquely defined in $ \mathcal{P}([\ell]) \cong \mathcal{P}(\mathbf{K}^\mathbf{1}) $ due to the fact that $ K^\perp = \{ 0 \} $ and $ \{ 0 \}^\perp = K $, and $ K $ and $ \{ 0 \} $ are the only two $ K $-linear subspaces of $ K $, when $ K $ is a field.

We now introduce the concept of sum-rank supports of vectors in $ \mathbb{F}^n $.

\begin{definition} [\textbf{Sum-rank supports}] \label{def sum-rank supports}
Let $ \mathbf{c} = (\mathbf{c}^{(1)}, \mathbf{c}^{(2)}, \ldots, $ $ \mathbf{c}^{(\ell)}) \in \mathbb{F}^n $, where $ \mathbf{c}^{(i)} \in \mathbb{F}^{n_i} $, for $ i = 1,2, \ldots, \ell $. We define its sum-rank support as the list 
$$ {\rm Supp}(\mathbf{c}) = (\mathcal{L}_1, \mathcal{L}_2, \ldots, \mathcal{L}_\ell) \in \mathcal{P}(\mathbf{K}^\mathbf{n}) , $$
where $ \mathcal{L}_i \subseteq K_i^{n_i} $ is the ($ K_i $-linear) row space of $ M_{\mathcal{A}_i}(\mathbf{c}^{(i)}) \in K_i^{m_i \times n_i} $, that is $ \mathcal{L}_i = {\rm Row}(M_{\mathcal{A}_i}(\mathbf{c}^{(i)})) $, for $ i = 1,2, \ldots, \ell $. 
\end{definition}

Observe that $ {\rm Row}(M_{\mathcal{A}_i}(\mathbf{c}^{(i)})) $ does not depend on $ \mathcal{A}_i $, for $ i = 1,2, \ldots, \ell $. Moreover, it follows from the definitions that
\begin{equation}
{\rm wt}_{SR}(\mathbf{c}) = \sum_{i=1}^\ell {\rm Rk}_{K_i}(M_{\mathcal{A}_i}(\mathbf{c}^{(i)})) = \sum_{i=1}^\ell \dim_{K_i}(\mathcal{L}_i) = {\rm Rk}({\rm Supp}(\mathbf{c})).
\end{equation}

By the discussion prior to Definition \ref{def sum-rank supports}, our definition of support extends that of Hamming supports as considered usually in the literature, by setting $ \mathbf{n} = \mathbf{1} $. Furthermore, describing supports as lists emphasizes the fact that different positions represent different slots in a time axis, as explained in \cite[Sec. IV]{forney}. 

In the case $ \ell = 1 $, our definition of support recovers that of rank supports as defined in \cite[Def. 2.1]{slides}, which has proven to be the right notion of rank support.

\section{Sum-rank support spaces} \label{sec sum-rank support spaces}

In this section, we introduce the concept of sum-rank support spaces (Subsection \ref{subsec equivalent defs sum-rank support spaces}) and give several characterizations of them. We will then show that the lattice of sum-rank support spaces is isomorphic to that of lists of vector spaces (Subsection \ref{subsec properties sum-rank support spaces}). We conclude by defining and studying restriction and shortening of linear codes (Subsection \ref{subsec restricted and shortened}), which naturally extend the corresponding operations of restriction and shortening for the Hamming and rank metrics.

\subsection{Definition and characterizations} \label{subsec equivalent defs sum-rank support spaces}

In this subsection, we give several equivalent definitions of sum-rank support spaces. As we will see (for instance, in Item 5 in Theorem \ref{th charact sum-rank support spaces}), these vector spaces behave as ambient spaces for the sum-rank metric. 

\begin{definition} [\textbf{Sum-rank support spaces}] \label{def sum-rank supp spaces}
Let $ \boldsymbol{\mathcal{L}} \in \mathcal{P}(\mathbf{K}^\mathbf{n}) $. We define the sum-rank support space in $ \mathbb{F}^n $ associated to $ \boldsymbol{\mathcal{L}} $ as the vector space
$$ \mathcal{V}_{\boldsymbol{\mathcal{L}}} = \{ \mathbf{c} \in \mathbb{F}^n \mid {\rm Supp}(\mathbf{c}) \subseteq \boldsymbol{\mathcal{L}} \}. $$
We denote by $ \mathcal{P}_{SR}( \mathbb{F}^n ) $ the set of sum-rank support spaces in $ \mathbb{F}^n $. 
\end{definition}

Observe that, in the case $ \ell = 1 $ and denoting $ K = K_1 $ and $ \mathcal{A} = \mathcal{A}_1 $, we obtain
\begin{equation}
\mathcal{V}_\mathcal{L} = \{ \mathbf{c} \in \mathbb{F}^n \mid {\rm Row}(M_{\mathcal{A}}(\mathbf{c})) \subseteq \mathcal{L} \},
\label{eq def rank support space}
\end{equation}
for a vector space $ \mathcal{L} \subseteq K^n $. Thus we recover rank support spaces as considered in \cite[Def. 3.2]{slides}, \cite[Def. 7]{rgmw} and \cite[Not. 25]{ravagnaniweights}. We recover Hamming support spaces (see \cite[Sec. II]{forney} for instance) when $ \mathbf{n} = \mathbf{1} $, by the discussions in Subsection \ref{subsec lattice unions spaces}.

We now give the following characterizations of sum-rank support spaces. In particular, by Item 3, sum-rank support spaces are indeed vector spaces (over $ \mathbb{F} $). 

\begin{theorem}\label{th charact sum-rank support spaces}
Let $ \mathcal{V} \subseteq \mathbb{F}^n $ be an arbitrary set. The following are equivalent:
\begin{enumerate}
\item
$ \mathcal{V} \in \mathcal{P}_{SR}(\mathbb{F}^n) $, that is, $ \mathcal{V} $ is a sum-rank support space.
\item
There exists $ \boldsymbol{\mathcal{L}} = (\mathcal{L}_1, \mathcal{L}_2, \ldots, \mathcal{L}_\ell) \in \mathcal{P}(\mathbf{K}^\mathbf{n}) $ such that 
$$ \mathcal{V} = \mathcal{V}_{\mathcal{L}_1} \times \mathcal{V}_{\mathcal{L}_2} \times \cdots \times \mathcal{V}_{\mathcal{L}_\ell}. $$
\item
There exist matrices $ A_i \in K_i^{N_i \times n_i} $ (possibly $ A_i = * \in K_i^{0 \times n_i} $), for $ i = 1,2, \ldots, \ell $, such that $ \mathcal{V} $ is the row vector space in $ \mathbb{F}^n $ of the block-diagonal matrix
$$ A = {\rm diag}(A_1, A_2, \ldots, A_\ell) \in \mathbb{F}^{N \times n}, $$
where $ N = N_1 + N_2 + \cdots + N_\ell $ and $ 0 \leq N_i \leq n_i $, for $ i = 1,2, \ldots, \ell $.
\item
$ \mathcal{V} $ is a vector space with a basis of vectors of sum-rank weight $ 1 $.
\item
$ \mathcal{V} $ is a vector space and there exists $ N = N_1 + N_2 + \cdots + N_\ell $, with $ 0 \leq N_i \leq n_i $, for $ i = 1,2,\ldots, \ell $, and a bijective linear sum-rank isometry $ \phi : \mathbb{F}^N \longrightarrow \mathcal{V} $, where the sum-rank metric in $ \mathbb{F}^N $ corresponds to the previous partition of $ N $.
\item
Define $ M(\mathcal{V}) = \{ (M_{\mathcal{A}_1} (\mathbf{c}^{(1)}), M_{\mathcal{A}_2} (\mathbf{c}^{(2)}), \ldots, M_{\mathcal{A}_\ell} (\mathbf{c}^{(\ell)}) ) \mid (\mathbf{c}^{(1)}, \mathbf{c}^{(2)}, \ldots, \mathbf{c}^{(\ell)}) \in \mathcal{V} \} $ and define the cartesian-product ring $ \mathcal{R} = K_1^{m_1 \times m_1} \times K_2^{m_2 \times m_2} \times \cdots \times K_\ell^{m_\ell \times m_\ell} $. Then $ M(\mathcal{V}) $ is a left $ \mathcal{R} $-module with component-wise matrix multiplication.
\end{enumerate}
\end{theorem}
\begin{proof}
The equivalence between Item 1 and Item 2 follows directly from the component-wise definition of inclusions in $ \mathcal{P}(\mathbf{K}^\mathbf{n}) $ (see Definition \ref{def basic operations lattice}).

Fix now $ i = 1,2, \ldots, \ell $. By \cite[Lemma 52]{rgmw}, a set $ \mathcal{W} \subseteq \mathbb{F}^{n_i} $ is a rank support space $ \mathcal{W} = \mathcal{V}_{\mathcal{L}_i} $ as in (\ref{eq def rank support space}), for some vector space $ \mathcal{L}_i \subseteq K_i^{n_i} $, if and only if, $ \mathcal{W} $ is Galois closed over $ K_i $. In other words, if it admits a basis of vectors in $ \mathbb{F}^{n_i} $ with components in $ K_i $. With this in mind, it is trivial to see that Item 2 and Item 3 are equivalent. 

Now, a vector $ \mathbf{c} = (\mathbf{c}^{(1)}, \mathbf{c}^{(2)}, \ldots, \mathbf{c}^{(\ell)}) \in \mathbb{F}^n $, with $ \mathbf{c}^{(i)} \in \mathbb{F}^{n_i} $, for $ i = 1,2, \ldots, \ell $, has sum-rank weight $ 1 $ if, and only if, there exists an index $ i $, $ \mathbf{d} \in K_i^{n_i} $ and $ \beta \in \mathbb{F}^* $ such that $ \mathbf{c}^{(i)} = \beta \mathbf{d} $, and $ \mathbf{c}^{(j)} = \mathbf{0} $ if $ j \neq i $. With this characterization of vectors of sum-rank weight $ 1 $, the equivalence between Item 3 and Item 4 is trivial.

Now we show that Item 2 implies Item 5. Define $ N_i = \dim_{K_i}(\mathcal{L}_i) $ and fix a basis $ \mathcal{B}_i $ of $ \mathcal{L}_i $, for $ i = 1,2, \ldots, \ell $. Define also $ N = N_1 + N_2 + \cdots + N_\ell $. We may now construct a linear sum-rank isometry $ \phi : \mathbb{F}^N \longrightarrow \mathcal{V} $ by sending the vectors of the canonical basis of $ \mathbb{F}^N $ corresponding to the $ i $th block of $ N_i $ coordinates to the vectors in $ \mathcal{B}_i $ positioned in the $ i $th block of $ n_i $ coordinates in $ \mathbb{F}^n $ (and zero elsewhere). The map $ \phi $ is a sum-rank isometry since it consists on multiplying on the right by the matrix
$$ A = {\rm diag}(A_1, A_2, \ldots, A_\ell) \in \mathbb{F}^{N \times n}, $$
where $ A_i \in K_i^{N_i \times n_i} $ is the full-rank matrix whose $ j $th row is the $ j $th vector in $ \mathcal{B}_i $, for $ j = 1,2, \ldots, N_i $ and for $ i = 1,2, \ldots, \ell $.

Assume now Item 5. We may construct a basis of $ \mathcal{V} $ of vectors of sum-rank weight $ 1 $ simply by taking the images by $ \phi $ of the canonical basis of $ \mathbb{F}^N $. Hence Item 4 follows.

We conclude by showing the equivalence between Item 2 and Item 6. Fix $ i = 1,2, \ldots, \ell $. In \cite[Appendix D]{rgmw}, it is shown that a set $ \mathcal{V}_i \subseteq \mathbb{F}^{n_i} $ is of the form $ \mathcal{V}_i = \mathcal{V}_{\mathcal{L}_i} $, for some vector space $ \mathcal{L}_i \subseteq K_i^{n_i} $, if and only if, $ M_{\mathcal{A}_i}(\mathcal{V}_i) \subseteq K_i^{m_i \times n_i} $ is a left module over the ring $ \mathcal{R}_i = K_i^{m_i \times m_i} $. The equivalence between Item 2 and Item 6 follows directly from this fact and the component-wise definition of multiplication in the cartesian-product ring $ \mathcal{R} = \mathcal{R}_1 \times \mathcal{R}_2 \times \cdots \times \mathcal{R}_\ell $.
\end{proof}

These characterizations naturally recover well-known characterizations of Ham\-ming-metric and rank-metric support spaces. Item 2 in the case $ \mathbf{n} = \mathbf{1} $ is the well-known fact that Hamming-support spaces are cartesian products of factors $ \{ 0 \} $ or $ \mathbb{F} $. Item 3 in the case $ \ell = 1 $ corresponds to the characterization in \cite[Th. 4.3 \& Prop. 5.5]{slides} or \cite[Lemma 52]{rgmw} of rank-metric support spaces as Galois closed spaces. In the case $ \mathbf{n} = \mathbf{1} $, it corresponds to the fact that Hamming-metric support spaces are generated by a subset of vectors in the canonical basis of $ \mathbb{F}^\ell $. 

Item 5 is typically used to motivate the effective length and degenerateness of linear codes, and justifies seeing support spaces as ambient spaces (see \cite[Sec. II]{forney} when $ \mathbf{n} = \mathbf{1} $). This was studied in the case $ \ell=1 $ in \cite[Sec. 6]{slides} and \cite[Subsec. IV-B]{similarities}. Item 5 can also be seen as MacWilliams' extension theorem \cite{macwilliams-thesis} for sum-rank support spaces (see also \cite[Subsec. IV-A]{similarities} for the case $ \ell = 1 $).

Finally, Item 6 is an arithmetic definition of these metrics, in terms of coordinate-wise matrix multiplication, that has been extensively used when $ \mathbf{n} = \mathbf{1} $, mostly in connection with evaluation codes, decoding and computation. Some instances are error-correcting pairs \cite{unified, on-the-existence}, BCH and Hartmann-Tzeng bounds and decoding of cyclic codes \cite[Sec. 3]{on-the-existence}, Feng-Rao or order bounds \cite{feng-rao}, attacks on McEliece-type cryptosystems \cite{attack-AG}, secure multiparty computation \cite{secure-computation} and private information retrieval \cite{private-computation}, among many others (see also the references inside the previous ones). Related works in the case $ \ell = 1 $ include \cite{greferath, RECP}. The example in \cite[Ex. 4.7]{greferath} is essentially the equivalence between Items 5 and 6 when $ \ell = 1 $, and states that MacWilliams' extension theorem \cite{macwilliams-thesis} holds for ($ N $-dimensional) rank support spaces, seen as left modules over the ring $ K_1^{N \times N} $. The work \cite{RECP} introduces rank error-correcting pairs. An interpretation of the attacks of the McEliece-type cryptosystems based on Gabidulin codes \cite[Sec. 5]{overbeck} can be made as in \cite{attack-AG}, since Gabidulin codes are of evaluation type and matrix multiplication translates into the Fr{\"o}benius morphism \cite[Prop. 1]{RECP}. It is worth recalling that linearized Reed-Solomon codes \cite{linearizedRS} are also obtained by a certain evaluation of skew polynomials and, generally, coordinate-wise matrix multiplications correspond to products of skew polynomials after such evaluations, as recently shown in \cite[Th. 2]{pirlrc}. For all these reasons, we have included Item 6 in the theorem, although we will not use it in the rest of the paper.

\subsection{Lattice properties of sum-rank support spaces} \label{subsec properties sum-rank support spaces}

As is the case for the Hamming and rank metrics, sum-rank support spaces behave as ambient spaces for sum-rank codes that can be attached bijectively to the objects in the support lattice (Subsection \ref{subsec lattice unions spaces}). The preservation of the lattice structure and duality by this bijection is an important tool for many coding-theoretic results when $ \ell = 1 $ and $ \mathbf{n} = \mathbf{1} $. As we will see, it is equally important in general for the sum-rank metric. We will implicitly use this isomorphism throughout the rest of the paper.

We start by defining sum-rank supports of subspaces, which constitutes an extension of Definition \ref{def sum-rank supports}.

\begin{definition} [\textbf{Sum-rank supports and weights of subspaces}] \label{def sum-rank supp of spaces}
Given a vector space $ \mathcal{D} \subseteq \mathbb{F}^n $, we define its sum-rank support as
$$ {\rm Supp}(\mathcal{D}) = \sum_{\mathbf{d} \in \mathcal{D}} {\rm Supp}(\mathbf{d}) \in \mathcal{P}(\mathbf{K}^\mathbf{n}). $$
We then define the sum-rank weight of $ \mathcal{D} $ as $ {\rm wt}_{SR}(\mathcal{D}) = {\rm Rk}({\rm Supp}(\mathcal{D})) $.
\end{definition}

Observe that, for all $ \mathbf{c} \in \mathbb{F}^n $, we have that $ {\rm Supp}(\mathbf{c}) = {\rm Supp}(\langle \mathbf{c} \rangle) $ and $ {\rm wt}_{SR}(\mathbf{c}) = {\rm wt}_{SR}(\langle \mathbf{c} \rangle) $, where $ \mathcal{D} = \langle \mathbf{c} \rangle \subseteq \mathbb{F}^n $ denotes the vector space generated by $ \mathbf{c} $.

In the case $ \mathbf{n} = \mathbf{1} $, sums of lists of subspaces coincide with unions of sets. Hence we recover in that case the definition of Hamming support and Hamming weight of a vector space \cite[Sec. II]{wei}. We recall that the original definition of support of a vector space goes back to \cite{helleseth79}. In the case $ \ell = 1 $, we recover the definition of rank support and rank weight of a vector space from \cite[Def. 2.1]{slides}.

With this definition, we may now show that the lattices of sum-rank supports and sum-rank support spaces are isomorphic. 

\begin{proposition} \label{prop lattice isom sum-rank support spaces}
The set $ \mathcal{P}_{SR}( \mathbb{F}^n ) $ is a sublattice of the lattice of vector subspaces of $ \mathbb{F}^n $. Moreover, the map
\begin{equation*}
\begin{array}{ccc}
 \mathcal{P}(\mathbf{K}^\mathbf{n}) & \longrightarrow & \mathcal{P}_{SR}( \mathbb{F}^n ) \\
 \boldsymbol{\mathcal{L}} & \mapsto & \mathcal{V}_{\boldsymbol{\mathcal{L}}} \\
\end{array}
\end{equation*}
is a lattice isomorphism, whose inverse is given by the support map
\begin{equation*}
\begin{array}{ccc}
 \mathcal{P}_{SR}( \mathbb{F}^n ) & \longrightarrow & \mathcal{P}(\mathbf{K}^\mathbf{n}) \\
 \mathcal{V} & \mapsto & {\rm Supp}(\mathcal{V}). \\
\end{array}
\end{equation*}
\end{proposition}
\begin{proof}
We first show that the given maps are inverse of each other. 

Let $ \boldsymbol{\mathcal{L}} \in \mathcal{P}(\mathbf{K}^\mathbf{n}) $, and let $ \boldsymbol{\mathcal{L}}^\prime = {\rm Supp}(\mathcal{V}_{\boldsymbol{\mathcal{L}}}) $. If $ \mathbf{c} \in \mathcal{V}_{\boldsymbol{\mathcal{L}}} $, then by definition, we have that $ {\rm Supp}(\mathbf{c}) \subseteq \boldsymbol{\mathcal{L}} $. Again by definition, we have that
$$ \boldsymbol{\mathcal{L}}^\prime = \sum_{\mathbf{c} \in \mathcal{V}_{\boldsymbol{\mathcal{L}}}} {\rm Supp}(\mathbf{c}) \subseteq \boldsymbol{\mathcal{L}}. $$
Now fix $ i = 1,2, \ldots, \ell $. There exist $ \mathbf{c}_{i,1}, \mathbf{c}_{i,2}, \ldots, \mathbf{c}_{i,N_i} \in \mathcal{V}_{\boldsymbol{\mathcal{L}}} $ such that
$$ (\{ \mathbf{0} \}, \ldots, \{ \mathbf{0} \}, \mathcal{L}_i, \{ \mathbf{0} \}, \ldots, \{ \mathbf{0} \}) \subseteq \sum_{j=1}^{N_i} {\rm Supp}(\mathbf{c}_{i,j}), $$
being $ \mathcal{L}_i $ placed in the $ i $th position. To see this, take a basis $ \mathbf{d}_{i,1}, \mathbf{d}_{i,2}, \ldots, \mathbf{d}_{i,N_i} \in K_i^{n_i} $ of $ \mathcal{L}_i $, and define $ \mathbf{c}_{i,j} \in \mathbb{F}^n $ as the vector that is identically zero except in the $ i $th block of $ n_i $ coordinates, where it is defined as $ \mathbf{d}_{i,j} $, for $ j = 1,2, \ldots, N_i $. Therefore, we conclude that
$$ \boldsymbol{\mathcal{L}} \subseteq \sum_{i=1}^\ell \sum_{j=1}^{N_i} {\rm Supp}(\mathbf{c}_{i,j}) \subseteq \boldsymbol{\mathcal{L}}^\prime, $$
and we deduce that $ \boldsymbol{\mathcal{L}} = {\rm Supp}(\mathcal{V}_{\boldsymbol{\mathcal{L}}}) $.

Let now $ \mathcal{V} \in \mathcal{P}_{SR}( \mathbb{F}^n ) $. By definition, there exists $ \boldsymbol{\mathcal{L}} \in \mathcal{P}(\mathbf{K}^\mathbf{n}) $ such that $ \mathcal{V} = \mathcal{V}_{\boldsymbol{\mathcal{L}}} $. By the previous paragraph, we know that $ {\rm Supp}(\mathcal{V}) = {\rm Supp}(\mathcal{V}_{\boldsymbol{\mathcal{L}}}) = \boldsymbol{\mathcal{L}} $. Thus $ \mathcal{V}_{{\rm Supp}(\mathcal{V})} = \mathcal{V} $. This concludes the proof that the given maps are inverse of each other.

Next, since $ \boldsymbol{\mathcal{L}}, \boldsymbol{\mathcal{L}}^\prime \subseteq \boldsymbol{\mathcal{L}} + \boldsymbol{\mathcal{L}}^\prime $, and $ \mathcal{V}_{\boldsymbol{\mathcal{L}} + \boldsymbol{\mathcal{L}}^\prime} $ is a vector space, we deduce that $ \mathcal{V}_{\boldsymbol{\mathcal{L}}} + \mathcal{V}_{\boldsymbol{\mathcal{L}}^\prime} \subseteq \mathcal{V}_{\boldsymbol{\mathcal{L}} + \boldsymbol{\mathcal{L}}^\prime} $. For the reversed inclusion, if $ \mathbf{c} = (\mathbf{c}^{(1)}, \mathbf{c}^{(2)}, \ldots, $ $ \mathbf{c}^{(\ell)}) $ $ \in \mathcal{V}_{\boldsymbol{\mathcal{L}} + \boldsymbol{\mathcal{L}}^\prime} $, we see that $ \mathbf{c}^{(i)} = \mathbf{d}^{(i)} + \mathbf{e}^{(i)} $, for some $ \mathbf{d}^{(i)} \in \mathcal{V}_{\mathcal{L}_i} $ and $ \mathbf{e}^{(i)} \in \mathcal{V}_{\mathcal{L}_i^\prime} $, for $ i = 1,2, \ldots, \ell $. Therefore $ \mathcal{V}_{\boldsymbol{\mathcal{L}} + \boldsymbol{\mathcal{L}}^\prime} \subseteq \mathcal{V}_{\boldsymbol{\mathcal{L}}} + \mathcal{V}_{\boldsymbol{\mathcal{L}}^\prime} $ and equality holds. Since $ \mathcal{V}_{\boldsymbol{\mathcal{L}}} \cap \mathcal{V}_{\boldsymbol{\mathcal{L}}^\prime} = \mathcal{V}_{\boldsymbol{\mathcal{L}} \cap \boldsymbol{\mathcal{L}}^\prime} $ is trivial, we conclude that the previous maps are lattice isomorphisms and $ \mathcal{P}_{SR}( \mathbb{F}^n ) $ is a sublattice of the lattice of vector subspaces of $ \mathbb{F}^n $.
\end{proof} 

In particular, we may give the following list of properties of sum-rank supports and sum-rank support spaces.

\begin{corollary} \label{cor lattice properties of support spaces}
Given $ \boldsymbol{\mathcal{L}}, \boldsymbol{\mathcal{L}}^\prime \in \mathcal{P}(\mathbf{K}^\mathbf{n}) $, the following properties hold:
\begin{enumerate}
\item
$ \dim(\mathcal{V}_{\boldsymbol{\mathcal{L}}}) = {\rm Rk}(\boldsymbol{\mathcal{L}}) $.
\item
$ \mathcal{V}_{\boldsymbol{\mathcal{L}}} \subseteq \mathcal{V}_{\boldsymbol{\mathcal{L}}^\prime} $ if, and only if, $ \boldsymbol{\mathcal{L}} \subseteq \boldsymbol{\mathcal{L}}^\prime $. 
\item
$ \mathcal{V}_{\boldsymbol{\mathcal{L}}} + \mathcal{V}_{\boldsymbol{\mathcal{L}}^\prime} = \mathcal{V}_{\boldsymbol{\mathcal{L}} + \boldsymbol{\mathcal{L}}^\prime} $.
\item
$ \mathcal{V}_{\boldsymbol{\mathcal{L}}} \cap \mathcal{V}_{\boldsymbol{\mathcal{L}}^\prime} = \mathcal{V}_{\boldsymbol{\mathcal{L}} \cap \boldsymbol{\mathcal{L}}^\prime} $.
\item
$ \mathcal{V}_{\boldsymbol{\mathcal{L}}}^\perp = \mathcal{V}_{\boldsymbol{\mathcal{L}}^\perp} $.
\item
$ \mathcal{V}_{\boldsymbol{\mathcal{T}}} = \mathcal{V}_{(K_1^{n_1}, K_2^{n_2}, \ldots, K_\ell^{n_\ell}) } = \mathbb{F}^n $ and $ \mathcal{V}_\mathbf{0} = \mathcal{V}_{(\{ 0 \}, \{ 0 \}, \ldots, \{ 0 \})} = \{ \mathbf{0} \} $.
\item
$ \mathcal{V}_{\boldsymbol{\mathcal{L}}^{\prime \prime}} = \mathcal{V}_{\boldsymbol{\mathcal{L}}} \oplus \mathcal{V}_{\boldsymbol{\mathcal{L}}^\prime} $ if, and only if, $ \boldsymbol{\mathcal{L}}^{\prime \prime} = \boldsymbol{\mathcal{L}} \oplus \boldsymbol{\mathcal{L}}^\prime $.
\end{enumerate}
\end{corollary}
\begin{proof}
All items follow directly from Proposition \ref{prop lattice isom sum-rank support spaces} and its proof, except Item 1 and Item 5. These items follow directly from Item 2 in Theorem \ref{th charact sum-rank support spaces} and the facts that
$$ \dim(\mathcal{V}_{\mathcal{L}_i}) = \dim_{K_i}(\mathcal{L}_i) \quad \textrm{and} \quad \mathcal{V}_{\mathcal{L}_i}^\perp = \mathcal{V}_{\mathcal{L}_i^\perp}, $$
respectively, for a vector space $ \mathcal{L}_i \subseteq K_i^{n_i} $, for $ i = 1,2, \ldots, \ell $. 
\end{proof}

We may also conclude the following, which intuitively says that sum-rank support spaces are the smallest ambient spaces of sum-rank linear codes.

\begin{corollary} \label{cor min support space containing}
If $ \mathcal{D} \subseteq \mathbb{F}^n $ is a vector space and $ \boldsymbol{\mathcal{L}} = {\rm Supp}(\mathcal{D}) $, then $ \mathcal{V}_{\boldsymbol{\mathcal{L}}} $ is the smallest sum-rank support space containing $ \mathcal{D} $.
\end{corollary}

\subsection{Restriction, shortening, change of bases and pre-shortening} \label{subsec restricted and shortened}

In this subsection, we introduce four basic operations on linear codes that are intimately related: restriction, shortening, change of bases and pre-shortening. We estimate the parameters of the linear codes obtained by such operations and establish the duality of restriction and shortening.

We start with changes of bases. They form a basic family of linear sum-rank isometries that correspond bijectively to the change-of-P-basis isometries \cite[Def. 12]{linearizedRS} for the skew metric given in \cite[Def. 9]{linearizedRS} (see \cite[Th. 2]{linearizedRS} for the exact connection). Actually, the terminology is inspired by the fact that what these isometries do is changing the bases, of each conjugacy class, on which we evaluate the linearized polynomials from \cite[Def. 20]{linearizedRS}. See also Appendix \ref{app} for some of the concepts regarding the skew metric.

\begin{definition}[\textbf{Change of bases}] \label{def change of bases}
Given invertible matrices $ B_i \in K_i^{n_i \times n_i} $, for $ i = 1,2, \ldots, \ell $, we define the corresponding change of bases as the map $ \pi_B : \mathbb{F}^n \longrightarrow \mathbb{F}^n $, where $ \pi_B(\mathbf{c}) = \mathbf{c}B^T $, and $ B = {\rm diag}(B_1, B_2, \ldots, B_\ell) \in \mathbb{F}^{n \times n} $. For a linear code $ \mathcal{C} \subseteq \mathbb{F}^n $, we define the corresponding change-of-bases code as $ \mathcal{C} B^T = \pi_B(\mathcal{C}) \subseteq \mathbb{F}^n $.
\end{definition}

Except for multiplication with an element in $ \mathbb{F}^* $, changes of bases are exactly all linear rank isometries when $ \ell = 1 $ by \cite[Th. 1]{berger}. In the case $ \mathbf{n} = \mathbf{1} $, changes of bases are just monomial maps (multiplication on the right by an invertible diagonal matrix), which except for permutation of coordinates, constitute also all linear Hamming isometries \cite{macwilliams-thesis}. It is not difficult to show that all linear sum-rank isometries are also changes of bases, after certain multiplications of elements in $ \mathbb{F}^* $ and permutations of blocks of coordinates (corresponding to time slots).

We now define pre-shortening, which as the name suggests, will be used to define shortening.

\begin{definition}[\textbf{Pre-shortening}]
Given a linear code $ \mathcal{C} \subseteq \mathbb{F}^n $ and $ \boldsymbol{\mathcal{L}} \in \mathcal{P}(\mathbf{K}^\mathbf{n}) $, we define its pre-shortening over $ \boldsymbol{\mathcal{L}} $ as the linear code $ \mathcal{C} \cap \mathcal{V}_{\boldsymbol{\mathcal{L}}} \subseteq \mathbb{F}^n $.
\end{definition}

We may now define restriction and shortening, which recover the classical notions in the cases $ \mathbf{n} = \mathbf{1} $ and $ \ell = 1 $. In the case $ \ell = 1 $, the definitions in \cite[Def. 3.2]{ravagnani-siam} and \cite[Def. 11]{similarities} are essentially equivalent. However, we adopt the approach in \cite{similarities}, which gives one restriction and one shortening for each support rather than each matrix, since different generator matrices of the same support give equivalent codes. In this way, the properties of the restricted and shortened codes are related to properties of the corresponding supports.

\begin{definition}[\textbf{Restriction and shortening}] \label{def restriction and shortening}
Let $ \mathcal{C} \subseteq \mathbb{F}^n $ be a linear code and let $ \boldsymbol{\mathcal{L}} \in \mathcal{P}(\mathbf{K}^\mathbf{n}) $. Choose $ \boldsymbol{\mathcal{L}}^\prime \in \mathcal{P}(\mathbf{K}^\mathbf{n}) $ such that $ \boldsymbol{\mathcal{T}} = \boldsymbol{\mathcal{L}}^\prime \oplus \boldsymbol{\mathcal{L}}^\perp $ (Definition \ref{def basic operations lattice}). Fix full-rank generator matrices $ A_i, A_i^\prime \in K_i^{N_i \times n_i} $ of $ \mathcal{L}_i, \mathcal{L}_i^\prime \subseteq K_i^{n_i} $, respectively, for $ i = 1,2, \ldots, \ell $. We define the restricted and shortened linear codes of $ \mathcal{C} $ over $ \boldsymbol{\mathcal{L}} $ as
$$ \mathcal{C}_{\boldsymbol{\mathcal{L}}} = \mathcal{C} A^T \subseteq \mathbb{F}^N \quad \textrm{and} \quad \mathcal{C}^{\boldsymbol{\mathcal{L}}} = (\mathcal{C} \cap \mathcal{V}_{\boldsymbol{\mathcal{L}}}) A^{\prime T} \subseteq \mathbb{F}^N, $$
respectively, where $ A = {\rm diag}(A_1, A_2, \ldots, A_\ell) \in \mathbb{F}^{N \times n} $, $ A^\prime = {\rm diag}(A_1^\prime, A_2^\prime, \ldots, A_\ell^\prime) \in \mathbb{F}^{N \times n} $ and $ N = N_1 + N_2 + \cdots + N_\ell = {\rm Rk}(\boldsymbol{\mathcal{L}}) = {\rm Rk}(\boldsymbol{\mathcal{L}}^\prime) $.
\end{definition}

The intuitive idea behind restriction and shortening is to obtain new linear codes of shorter length, obtained by a given family of linear projections, where information on the sum-rank properties of the new codes can be derived from the supports associated to such linear projections. In the case $ \mathbf{n} = \mathbf{1} $, the allowed linear projections are projections over a subset of coordinates (after possibly applying a monomial map), whereas in the case $ \ell = 1 $, linear projections are arbitrary but with coefficients over the subfield $ K = K_1 $. 

In the case of shortening, applying a pre-shortening guarantees that the minimum sum-rank distance of the shortened code is not smaller than that of the original code (Corollary \ref{cor parameters restricted and shortened}). In this regard, the definition using $ \boldsymbol{\mathcal{L}}^\prime $ might seem unnatural, but it will allow us to show that pre-shortened and shortened codes are canonically sum-rank isometric, and that shortening and restriction are dual operations. 

Note that the definitions of restriction and shortening depend on the support $ \boldsymbol{\mathcal{L}}^\prime $ and the generator matrices of the subspaces $ \mathcal{L}_i $ and $ \mathcal{L}_i^\prime $. However, there is a canonical linear sum-rank isometry between any two restricted codes or any two shortened codes over the same support $ \boldsymbol{\mathcal{L}} \in \mathcal{P}(\mathbf{K}^\mathbf{n}) $, which justifies disregarding the dependency on generator matrices and $ \boldsymbol{\mathcal{L}}^\prime $. Before proving this, we need the following tool from \cite[Prop. 12]{similarities}.

\begin{lemma}[\textbf{\cite{similarities}}] \label{lemma basis of L prime}
Fix a field $ K $ and $ N $-dimensional subspaces $ \mathcal{L}, \mathcal{L}^\prime \in \mathcal{P}(K^n) $ such that $ \mathcal{L} $ is generated by $ A \in K^{N \times n} $. It holds that $ \mathcal{L}^\prime \oplus \mathcal{L}^\perp = K^{n} $ if, and only if, $ \mathcal{L}^\prime $ has a generator matrix $ A^\prime \in K^{N \times n} $ such that $ A^\prime A^T = I_N $.
\end{lemma}

The next two lemmas, being the first one trivial, recover \cite[Lemma 7]{similarities} when $ \ell = 1 $. 

\begin{lemma}
Let the notation be as in Definition \ref{def restriction and shortening}. If $ \widetilde{A}_i \in K_i^{N_i \times n_i} $ is another generator matrix of $ \mathcal{L}_i \subseteq K_i^{n_i} $, then there exists a unique invertible matrix $ B_i \in K_i^{N_i \times N_i} $ such that $ \widetilde{A}_i = B_i A_i $, for $ i = 1,2, \ldots, \ell $. Therefore, it holds that
$$ \mathcal{C} A^T B^T = \mathcal{C} \widetilde{A}^T. $$
In other words, $ \mathcal{C} \widetilde{A}^T $ is a change of bases of the code $ \mathcal{C} A^T $. 

In particular, there is a canonical linear sum-rank isometry between any two restricted codes over the same sum-rank support.
\end{lemma}

\begin{lemma} \label{lemma projecting from pre-shortened}
Let the notation be as in Definition \ref{def restriction and shortening}. The map $ \pi_{A^\prime} : \mathcal{C} \cap \mathcal{V}_{\boldsymbol{\mathcal{L}}} \longrightarrow (\mathcal{C} \cap \mathcal{V}_{\boldsymbol{\mathcal{L}}}) A^{\prime T} $, given by $ \pi_{A^\prime}(\mathbf{c}) = \mathbf{c} A^{\prime T} $ is a linear sum-rank isometry. In particular, it holds that
\begin{equation}
 \dim(\mathcal{C}^{\boldsymbol{\mathcal{L}}}) = \dim(\mathcal{C} \cap \mathcal{V}_{\boldsymbol{\mathcal{L}}}),
\label{eq dimension of pre-shortened and shortened}
\end{equation}
and there is a canonical linear sum-rank isometry between any two shortened codes over the same sum-rank support.
\end{lemma}
\begin{proof}
By Lemma \ref{lemma basis of L prime}, we may take the matrices $ A_i \in K_i^{N_i \times n_i} $, for $ i = 1,2, \ldots, \ell $, satisfying $ A^\prime A^T = I $, where $ A = {\rm diag}(A_1, A_2, \ldots, A_\ell) \in \mathbb{F}^{N \times n} $ and $ A^\prime = {\rm diag}(A_1^\prime, A_2^\prime, \ldots, A_\ell^\prime) \in \mathbb{F}^{N \times n} $. If $ \mathbf{c} \in \mathcal{C} \cap \mathcal{V}_{\boldsymbol{\mathcal{L}}} $, then there exists $ \mathbf{x} \in \mathbb{F}^N $ such that $ \mathbf{c} = \mathbf{x} A $ by Item 3 in Theorem \ref{th charact sum-rank support spaces}. If $ \pi_{A^\prime}(\mathbf{c}) = \mathbf{0} $, then
$$ \mathbf{x} = \mathbf{x} (A A^{\prime T}) = \mathbf{c} A^{\prime T} = \mathbf{0}, $$
and then $ \mathbf{c} = \mathbf{x} A = \mathbf{0} $. This means that the map $ \pi_{A^\prime} $ is a vector space isomorphism.

Furthermore, since $ \pi_{A^\prime} (\mathbf{c}) = \mathbf{x} (A A^{\prime T}) = \mathbf{x} $, we have that 
$$ {\rm wt}_{SR}(\mathbf{c}) = {\rm wt}_{SR}(\mathbf{x}) = {\rm wt}_{SR}(\mathbf{x} (A A^{\prime T})) = {\rm wt}_{SR}(\pi_{A^\prime} (\mathbf{c})). $$
Therefore, the map $ \pi_{A^\prime} $ is a sum-rank isometry, and we are done.
\end{proof}

Our definition of restricted and shortened codes is not necessary in its full generality for the Hamming metric as we may always choose the basis $ 1 $ for $ K_i $ or $ 0 $ for $ \{ 0 \} $. For other choices, we obtain isometric codes by a monomial map. Furthermore, when $ \mathbf{n} = \mathbf{1} $, it must hold that $ \boldsymbol{\mathcal{L}}^\prime = \boldsymbol{\mathcal{L}} $, since duals and complementaries coincide and are uniquely defined in that case.

We relate now dimensions of restricted, shortened and pre-shortened codes. These relations recover Forney's duality lemmas \cite[Lemmas 1 \& 2]{forney} when $ \mathbf{n} = \mathbf{1} $. 

\begin{proposition} \label{prop dimensions restricted shortened}
Given a linear code $ \mathcal{C} \subseteq \mathbb{F}^n $ and $ \boldsymbol{\mathcal{L}} \in \mathcal{P}(\mathbf{K}^\mathbf{n}) $, it holds that
\begin{equation*}
\begin{split}
\dim(\mathcal{C}_{\boldsymbol{\mathcal{L}}}) & = {\rm Rk}(\boldsymbol{\mathcal{L}}) - \dim(\mathcal{C}^\perp \cap \mathcal{V}_{\boldsymbol{\mathcal{L}}}) = {\rm Rk}(\boldsymbol{\mathcal{L}}) - \dim((\mathcal{C}^\perp )^{\boldsymbol{\mathcal{L}}}) \\
 & = \dim(\mathcal{C}) - \dim(\mathcal{C} \cap \mathcal{V}_{\boldsymbol{\mathcal{L}}^\perp}) = \dim(\mathcal{C}) - \dim(\mathcal{C} ^{(\boldsymbol{\mathcal{L}}^\perp)}). 
\end{split}
\end{equation*}
\end{proposition}
\begin{proof}
Fix $ i = 1,2,\ldots,\ell $. Given $ \mathbf{c}^{(i)} \in \mathbb{F}^{n_i} $, write it as $ \mathbf{c}^{(i)} = \sum_{j=1}^{m_i} \alpha_j^{(i)} \mathbf{c}_j^{(i)} $, where $ \mathbf{c}_j^{(i)} \in K_i^{n_i} $, for $ j = 1,2, \ldots, m_i $, as in (\ref{eq def matrix representation map}). By Equation (\ref{eq def rank support space}), it is straightforward to check that $ \mathbf{c}^{(i)} A_i^T = \mathbf{0} $ if, and only if, $ \mathbf{c}^{(i)} \in \mathcal{V}_{\mathcal{L}_i^\perp} $. Hence by Definition \ref{def sum-rank supp spaces}, $ \mathbf{c} A^T = \mathbf{0} $ if, and only if, $ {\rm Supp}(\mathbf{c}) \subseteq \boldsymbol{\mathcal{L}}^\perp $, for $ \mathbf{c} \in \mathbb{F}^n $. Thus we have that
$$ \ker (\pi_A) = \mathcal{V}_{\boldsymbol{\mathcal{L}}^\perp}. $$
Therefore we conclude that
$$ \dim(\mathcal{C}_{\boldsymbol{\mathcal{L}}}) = \dim(\mathcal{C}) - \dim(\mathcal{C} \cap \ker(\pi_A)) = \dim(\mathcal{C}) - \dim(\mathcal{C} \cap \mathcal{V}_{\boldsymbol{\mathcal{L}}^\perp}) $$
by the first isomorphism theorem. Now, using the dimension formulas and using that $ \dim(\mathcal{V}_{\boldsymbol{\mathcal{L}}}) = {\rm Rk}(\boldsymbol{\mathcal{L}}) $ and $ \mathcal{V}_{\boldsymbol{\mathcal{L}}}^\perp = \mathcal{V}_{\boldsymbol{\mathcal{L}}^\perp} $ by Corollary \ref{cor lattice properties of support spaces}, we deduce that
$$ {\rm Rk}(\boldsymbol{\mathcal{L}}) - \dim(\mathcal{C}^\perp \cap \mathcal{V}_{\boldsymbol{\mathcal{L}}}) = \dim(\mathcal{C}) - \dim(\mathcal{C} \cap \mathcal{V}_{\boldsymbol{\mathcal{L}}^\perp}). $$
\end{proof}

We now show that restriction and shortening are dual operations. Observe however that this result requires using the right bases for $ \boldsymbol{\mathcal{L}} $ and $ \boldsymbol{\mathcal{L}}^\prime $. In the Hamming-metric case, this can be seen as using inverse monomial maps when defining restriction and shortening as multiplication by diagonal matrices.

\begin{corollary}
Let the notation be as in Definition \ref{def restriction and shortening}, and take full-rank generator matrices $ A_i, A_i^\prime \in K_i^{N_i \times n_i} $ of $ \mathcal{L}_i, \mathcal{L}_i^\prime \subseteq K_i^{n_i} $, respectively, such that $ A^\prime_i A_i^T = I $, for $ i = 1,2, \ldots, \ell $, which exist by Lemma \ref{lemma basis of L prime}. Then it holds that
$$ (\mathcal{C}_{\boldsymbol{\mathcal{L}}})^\perp \equiv (\mathcal{C} A^T)^\perp = (\mathcal{C}^\perp \cap \mathcal{V}_{\boldsymbol{\mathcal{L}}}) A^{\prime T} \equiv (\mathcal{C}^\perp)^{\boldsymbol{\mathcal{L}}} . $$
By exchanging $ \mathcal{C} $ and $ \mathcal{C}^\perp $, it also holds that $ (\mathcal{C}^\perp)_{\boldsymbol{\mathcal{L}}} = (\mathcal{C}^{\boldsymbol{\mathcal{L}}})^\perp $.
\end{corollary}
\begin{proof}
Take $ \mathbf{c} \in \mathcal{C} $ and $ \mathbf{d} \in \mathcal{C}^\perp \cap \mathcal{V}_{\boldsymbol{\mathcal{L}}} $. Since $ \mathbf{d} \in \mathcal{V}_{\boldsymbol{\mathcal{L}}} $, there exists $ \mathbf{x} \in \mathbb{F}^N $ such that $ \mathbf{d} = \mathbf{x} A $, where $ N = {\rm Rk}(\boldsymbol{\mathcal{L}}) $. Then it holds that
$$ (\mathbf{c} A^T) (\mathbf{d} A^{\prime T})^T = \mathbf{c} (A^T A^\prime A^T) \mathbf{x}^T = \mathbf{c} (\mathbf{x} A)^T = \mathbf{c} \mathbf{d}^T = \mathbf{0}. $$
Therefore, we have that $ (\mathcal{C} A^T)^\perp \subseteq (\mathcal{C}^\perp \cap \mathcal{V}_{\boldsymbol{\mathcal{L}}}) A^{\prime T} $. Now by computing dimensions using Proposition \ref{prop dimensions restricted shortened}, both are equal.
\end{proof}

We conclude by estimating the parameters of restricted and shortened codes. Constructing new codes from old is of special interest for the sum-rank metric, since not many constructions are known and they seem to not be straightforward, as mentioned in Section \ref{sec intro}. Better estimates on dimensions will be given in Proposition \ref{prop MSRD rank for intermediate info leakage}.

\begin{corollary} \label{cor parameters restricted and shortened}
Given a linear code $ \mathcal{C} \subseteq \mathbb{F}^n $ and $ \boldsymbol{\mathcal{L}} \in \mathcal{P}(\mathbf{K}^\mathbf{n}) $, the following hold:
\begin{enumerate}
\item
$ \dim(\mathcal{C}_{\boldsymbol{\mathcal{L}}}) \geq \dim(\mathcal{C}) - (n - {\rm Rk}(\boldsymbol{\mathcal{L}})) $ and $ {\rm d}_{SR}(\mathcal{C}_{\boldsymbol{\mathcal{L}}}) \geq {\rm d}_{SR}(\mathcal{C}) - (n - {\rm Rk}(\boldsymbol{\mathcal{L}})) $.
\item
$ \dim(\mathcal{C}^{\boldsymbol{\mathcal{L}}}) \geq \dim(\mathcal{C}) - (n - {\rm Rk}(\boldsymbol{\mathcal{L}})) $ and $ {\rm d}_{SR}(\mathcal{C}^{\boldsymbol{\mathcal{L}}}) \geq {\rm d}_{SR}(\mathcal{C}) $.
\end{enumerate}
\end{corollary}
\begin{proof}
The bounds on dimensions follow from Proposition \ref{prop dimensions restricted shortened}, the lower bound on $ {\rm d}_{SR}(\mathcal{C}_{\boldsymbol{\mathcal{L}}}) $ follows from the definitions, and that on $ {\rm d}_{SR}(\mathcal{C}^{\boldsymbol{\mathcal{L}}}) $ follows from Lemma \ref{lemma projecting from pre-shortened}.
\end{proof}

\section{Applications} \label{sec some applications}

In this section, we present several applications of sum-rank supports and support spaces. We will focus on generalized sum-rank weights (Subsection \ref{subsec gen sum-rank weights}), properties of and operations on maximum sum-rank distance (MSRD) codes (Subsection \ref{subsec MSRD codes}), and sum-rank effective length and degenerate codes (Subsection \ref{subsec sum-rank degenerate}).

\subsection{Generalized sum-rank weights} \label{subsec gen sum-rank weights}

In this subsection, we give the definition and main properties of (relative) generalized sum-rank weights. We will present equivalent definitions in terms of pre-shortened, shortened and restricted codes, and in terms of sum-rank weights of subspaces (Definition \ref{def sum-rank supp of spaces}). In the case $ \mathbf{n} = \mathbf{1} $, equivalent definitions exist seeing the underlying linear code as a projective system \cite[Sec. II]{tsfasman} or a matroid \cite{barg-matroid} (based on \cite[Th. 2]{wei}), or in terms of anticodes \cite{ravagnaniweights}, among others. The $ q $-analog of a matroid has been recently introduced in \cite{q-matroid}, where its connection with linear codes in the case $ \ell = 1 $ was given. Anticodes when $ \ell = 1 $ were used in \cite{oggier, ravagnaniweights}. Analogous reinterpretations of generalized sum-rank weights are left open. We will briefly discuss their application to measuring information leakage in multishot matrix-multiplicative wire-tap channels \cite{secure-multishot}.

We start by defining relative generalized sum-rank weights, and their dual notion.

\begin{definition} [\textbf{Relative generalized sum-rank weights}]
Given nested linear codes $ \mathcal{C}_2 \subsetneqq \mathcal{C}_1 \subseteq \mathbb{F}^n $, we define their $ r $th relative generalized sum-rank weight as
\begin{equation*}
\begin{split}
{\rm d}_{SR, r}(\mathcal{C}_1, \mathcal{C}_2) = \min \left\lbrace \right. & {\rm Rk}(\boldsymbol{\mathcal{L}}) \mid \boldsymbol{\mathcal{L}} \in \mathcal{P}(\mathbf{K}^\mathbf{n}), \textrm{ and} \\
& \left. \dim (\mathcal{C}_1 \cap \mathcal{V}_{\boldsymbol{\mathcal{L}}}) - \dim (\mathcal{C}_2 \cap \mathcal{V}_{\boldsymbol{\mathcal{L}}}) \geq r \right\rbrace,
\end{split}
\end{equation*}
for $ r = 1,2, \ldots, \dim(\mathcal{C}_1 / \mathcal{C}_2) $. We also define the parameter 
\begin{equation*}
\begin{split}
{\rm K}_{SR, \mu}(\mathcal{C}_1, \mathcal{C}_2) = \max \left\lbrace \right. & \dim (\mathcal{C}_1 \cap \mathcal{V}_{\boldsymbol{\mathcal{L}}}) - \dim (\mathcal{C}_2 \cap \mathcal{V}_{\boldsymbol{\mathcal{L}}}) \mid \\ & \boldsymbol{\mathcal{L}} \in \mathcal{P}(\mathbf{K}^\mathbf{n}), \textrm{ and } \left. {\rm Rk}(\boldsymbol{\mathcal{L}}) = \mu \right\rbrace,
\end{split}
\end{equation*}
for $ \mu = 0,1, \ldots, n $. For a single linear code $ \mathcal{C} \subseteq \mathbb{F}^n $, we define its $ r $th generalized sum-rank weight as $ {\rm d}_{SR, r}(\mathcal{C}) = {\rm d}_{SR, r}(\mathcal{C}, \{ \mathbf{0} \} ) $, and we define $ {\rm K}_{SR, \mu}(\mathcal{C}) = {\rm K}_{SR, \mu}(\mathcal{C}, \{ \mathbf{0} \} ) $, for $ r = 1,2, \ldots, \dim(\mathcal{C}) $ and $ \mu = 0,1, \ldots, n $.
\end{definition}

Since sum-rank support spaces extend both Hamming-support and rank-support spaces (Definition \ref{def sum-rank supp spaces}), we deduce automatically that generalized sum-rank weights particularize to generalized Hamming weights \cite{helleseth79, luo, wei} and generalized rank weights \cite{rgrw, oggier} when $ \mathbf{n} = \mathbf{1} $ and $ \ell = 1 $, respectively.

We have defined relative generalized sum-rank weights in terms of pre-shortened codes. This was the original approach in the case $ \ell = 1 $ \cite[Def. 2]{rgrw}. Just as in the case $ \mathbf{n} = \mathbf{1} $ (see \cite[Sec. 2]{forney}, \cite[Sec. 3]{luo} and \cite[Th. 2]{wei}), we may give equivalent definitions in terms of restricted and shortened codes. This is due to the following identities, which follow directly from Proposition \ref{prop dimensions restricted shortened}. Observe that restricted and shortened codes are the key description in the matroidal approach to generalized weights (see \cite{barg-matroid, q-matroid}).

\begin{proposition} \label{prop dimensions for equivalent gen weights}
For linear codes $ \mathcal{C}_2 \subsetneqq \mathcal{C}_1 \subseteq \mathbb{F}^n $ and for $ \boldsymbol{\mathcal{L}} \in \mathcal{P}(\mathbf{K}^\mathbf{n}) $, it holds that
\begin{equation*}
\begin{split}
\dim (\mathcal{C}_1 \cap \mathcal{V}_{\boldsymbol{\mathcal{L}}}) - \dim (\mathcal{C}_2 \cap \mathcal{V}_{\boldsymbol{\mathcal{L}}}) & = \dim (\mathcal{C}_1^{\boldsymbol{\mathcal{L}}}) - \dim (\mathcal{C}_2^{\boldsymbol{\mathcal{L}}}) \\
 & = \dim ((\mathcal{C}_2^\perp)_{\boldsymbol{\mathcal{L}}}) - \dim ((\mathcal{C}_1^\perp)_{\boldsymbol{\mathcal{L}}}). 
\end{split}
\end{equation*}
\end{proposition}

The use of nested linear code pairs is a usual technique to protect messages from both noise and information leakage to a wire-tapper. This technique goes back to \cite{ozarow} for $ \mathbf{n} = \mathbf{1} $, used in \cite{rgrw, silva-universal} for $ \ell = 1 $, and recently used in \cite{secure-multishot} for the general case. 

Informally, using nested linear codes $ \mathcal{C}_2 \subsetneqq \mathcal{C}_1 \subseteq \mathbb{F}^n $ for encoding as in \cite[Def. 3]{secure-multishot}, and integers $ \mu = 0,1,2, \ldots, n $ and $ r = 1,2, \ldots, \dim(\mathcal{C}_1 / \mathcal{C}_2) $, it holds that
\begin{enumerate}
\item
$ \mu = {\rm d}_{SR, r}(\mathcal{C}_2^\perp, \mathcal{C}_1^\perp) $ is the minimum number of links that an adversary needs to wire-tap in order to obtain at least $ r $ units of information (number of bits multiplied by $ \log_2|\mathbb{F}| $) of the secret message,
\item
$ r = {\rm K}_{SR, \mu}(\mathcal{C}_2^\perp, \mathcal{C}_1^\perp) $ is the maximum information (number of bits multiplied by $ \log_2|\mathbb{F}| $) about the secret message that can be obtained by wire-tapping at most $ \mu $ links of the network,
\end{enumerate}
on $ \ell $ shots of a linearly coded network, with $ n_i $ outgoing links in the $ i $th shot, that realizes a matrix-multiplicative wire-tap channel. This result follows directly from \cite[Lemma 1]{secure-multishot}. Further refinements as in \cite[Subsec. VII-A]{similarities} are left to the reader.

We now give the monotonicity properties of generalized sum-rank weights. The following result recovers \cite[Th. 6.1]{helleseth79}, \cite[Th. 1]{wei} and \cite[Prop. 1 \& 2]{luo} when $ \mathbf{n} = \mathbf{1} $, and it recovers \cite[Th. 1 \& Lemma 4]{rgrw} when $ \ell = 1 $.

\begin{lemma}[\textbf{Monotonicity}] \label{lemma monotonicity}
Given nested linear codes $ \mathcal{C}_2 \subsetneqq \mathcal{C}_1 \subseteq \mathbb{F}^n $ with $ k = \dim(\mathcal{C}_1/ \mathcal{C}_2) $, it holds that $ {\rm K}_{SR, 0}(\mathcal{C}_1, \mathcal{C}_2) = 0 $, $ {\rm K}_{SR, n}(\mathcal{C}_1, \mathcal{C}_2) = k $, $ {\rm d}_{SR, 1}(\mathcal{C}_1, \mathcal{C}_2) \geq 1 $, $ {\rm d}_{SR, k}(\mathcal{C}_1, \mathcal{C}_2) \leq n - \dim(\mathcal{C}_2) $,
$$ 0 \leq {\rm K}_{SR, \mu + 1}(\mathcal{C}_1, \mathcal{C}_2) - {\rm K}_{SR, \mu}(\mathcal{C}_1, \mathcal{C}_2) \leq 1, \textrm{ and} $$
$$ {\rm d}_{SR,r}(\mathcal{C}_1, \mathcal{C}_2) < {\rm d}_{SR,r+1}(\mathcal{C}_1, \mathcal{C}_2), $$
for $ \mu = 0,1,2, \ldots, n $ and $ r = 1,2, \ldots, k-1 $.
\end{lemma}
\begin{proof}
All statements are trivial from the definitions, except for the monotonicity of relative generalized sum-rank weights, which will be proven at the end of this Subsection, and the inequality $ {\rm d}_{SR, k}(\mathcal{C}_1, \mathcal{C}_2) \leq n - \dim(\mathcal{C}_2) $, which will be given in Proposition \ref{prop singleton bound}.
\end{proof}

We may connect both types of parameters by the following proposition, which recovers \cite[Th. 3]{luo} and \cite[Prop. 2]{luo} (see also \cite[Sec. 3]{forney}) when $ \mathbf{n} = \mathbf{1} $, and recovers \cite[Lemma 4]{rgrw} when $ \ell = 1 $.

\begin{lemma} \label{lemma connection between profiles}
Given nested linear codes $ \mathcal{C}_2 \subsetneqq \mathcal{C}_1 \subseteq \mathbb{F}^n $, it holds that
\begin{equation*}
\begin{split}
{\rm K}_{SR, \mu}(\mathcal{C}_1, \mathcal{C}_2) & = \max \{ r \in [k] \mid {\rm d}_{SR,r}(\mathcal{C}_1, \mathcal{C}_2) \leq \mu \}, \\
{\rm d}_{SR,r}(\mathcal{C}_1, \mathcal{C}_2) & = \min \{ \mu \in [n] \mid {\rm K}_{SR, \mu}(\mathcal{C}_1, \mathcal{C}_2) \geq r \} \\
& = \min \{ \mu \in [n] \mid {\rm K}_{SR, \mu}(\mathcal{C}_1, \mathcal{C}_2) = r \},
\end{split}
\end{equation*}
for $ \mu = 0, 1,2, \ldots, n $ and $ r = 1,2, \ldots, k $, where $ k = \dim(\mathcal{C}_1 / \mathcal{C}_2) $.
\end{lemma}
\begin{proof}
The first two equalities follow easily from the definitions and can be proven exactly as in \cite[Th. 2]{luo}. The last equality follows from the monotonicity and extremal properties of $ {\rm K}_{SR, \mu}(\mathcal{C}_1, \mathcal{C}_2) $ from Lemma \ref{lemma monotonicity}.
\end{proof}

We next give a description of relative generalized sum-rank weights in terms of sum-rank weights of subspaces of the corresponding dimension (recall Definition \ref{def sum-rank supp of spaces}). This is analogous to the original definition of generalized Hamming weights by Wei \cite{wei}. The case $ \ell = 1 $ was first given in \cite[Cor. 4.4]{slides} (see also \cite[Th. 3]{similarities} for relative weights). This result justifies the term \textit{generalized weights}.

\begin{proposition} \label{proposition as minimum rank weights}
Given nested linear codes $ \mathcal{C}_2 \subsetneqq \mathcal{C}_1 \subseteq \mathbb{F}^n $, it holds that
\begin{equation*}
\begin{split}
{\rm d}_{SR,r}(\mathcal{C}_1, \mathcal{C}_2) = \min \{ & {\rm wt}_{SR}(\mathcal{D}) \mid \mathcal{D} \subseteq \mathcal{C}_1, \mathcal{D} \cap \mathcal{C}_2 = \{ \mathbf{0} \}, \\
 & \dim(\mathcal{D}) = r \},
\end{split}
\end{equation*}
for $ r = 1,2, \ldots, \dim(\mathcal{C}_1 / \mathcal{C}_2) $. In particular, the first relative generalized sum-rank weight of the pair is its relative minimum sum-rank distance, given by
$$ {\rm d}_{SR} (\mathcal{C}_1, \mathcal{C}_2) = \min \{ {\rm wt}_{SR}(\mathbf{c}) \mid \mathbf{c} \in \mathcal{C}_1 \setminus \mathcal{C}_2 \}. $$
By choosing $ \mathcal{C}_2 = \{ \mathbf{0} \} $, the same result holds for generalized sum-rank weights and the minimum sum-rank distance.
\end{proposition}
\begin{proof}
With the definitions and results from Sections \ref{sec the sum-rank metric and support} and \ref{sec sum-rank support spaces}, the proof can be translated mutatis mutandis from those in \cite[Th. 3]{similarities} or \cite[Prop. 12]{rgmw}.
\end{proof}

We now exend Wei's duality theorem, given in \cite[Th. 3]{wei} when $ \mathbf{n} = \mathbf{1} $ and in \cite{jerome} when $ \ell = 1 $.

\begin{theorem}[\textbf{Wei duality}] \label{th GSRW duality}
Let $ \mathcal{C} \subseteq \mathbb{F}^n $ be a $ k $-dimensional linear code. If we denote $ d_r = {\rm d}_{SR,r}(\mathcal{C}) $, for $ r = 1,2, \ldots,k $, and $ d_s^\perp = {\rm d}_{SR,s}(\mathcal{C}^\perp) $, for $ s = 1,2, \ldots, n-k $, then it holds that
$$ \{ 1,2, \ldots, n \} = \{ d_1, d_2, \ldots, d_k \} \cup $$
$$ \{ n+1-d_1^\perp, n+1-d_2^\perp, \ldots, n+1-d_{n-k}^\perp \}, $$
where the union is disjoint. In particular, the generalized sum-rank weights of $ \mathcal{C} $ uniquely determine those of $ \mathcal{C}^\perp $.
\end{theorem}
\begin{proof}
With the definitions and results from Sections \ref{sec the sum-rank metric and support} and \ref{sec sum-rank support spaces} (Proposition \ref{prop dimensions restricted shortened} and Lemma \ref{lemma connection between profiles}), the proof can be translated mutatis mutandis from that in \cite[App. B]{similarities}.
\end{proof}

We conclude this section with a more novel result, which states that there exists a hierarchy of bounds for (relative) generalized sum-rank weights of linear codes, depending on how much one refines the partition $ n = n_1 + n_2 + \cdots + n_\ell $. The result states that bounds for finer partitions can be directly translated to bounds for less fine partitions. The extreme cases would be that of generalized Hamming weights ($ \mathbf{n} = \mathbf{1} $) and generalized rank weights ($ \ell = 1 $), and the corresponding result is then \cite[Th. 7]{similarities}. 

The main observation is that the sum-rank weight of a linear code can be expressed as the minimum of Hamming weights of all of its change-of-bases codes (Definition \ref{def change of bases}). This was proven in the case $ \ell = 1 $ in \cite[Th. 1]{similarities} for subspaces, for the skew metric and vectors in \cite[Prop. 14]{linearizedRS}, and for the sum-rank metric and vectors in the proof of \cite[Th. 3]{linearizedRS}. We now establish it for the sum-rank metric and arbitrary subspaces.

\begin{theorem} \label{th connection GHW and GSRW}
Let $ \mathcal{D} \subseteq \mathbb{F}^n $ be a vector space. It holds that
\begin{equation*}
\begin{split}
{\rm wt}_{SR}(\mathcal{D}) = \min \{ & {\rm wt}_H(\mathcal{D} A) \mid A = {\rm diag}(A_1, A_2, \ldots, A_\ell), \\
 & A_i \in K_i^{n_i \times n_i} \textrm{ invertible, for } i = 1,2, \ldots, \ell \}.
\end{split}
\end{equation*}
In particular, if $ \mathcal{C}_2 \subsetneqq \mathcal{C}_1 \subseteq \mathbb{F}^n $ are linear codes, then for all $ r = 1,2, \ldots, \dim(\mathcal{C}_1 / \mathcal{C}_2) $, it holds that
\begin{equation*}
\begin{split}
{\rm d}_{SR,r}(\mathcal{C}_1, \mathcal{C}_2) = \min \{ & {\rm d}_{H,r}(\mathcal{C}_1 A, \mathcal{C}_2 A) \mid A = {\rm diag}(A_1, A_2, \ldots, A_\ell), \\
 & A_i \in K_i^{n_i \times n_i} \textrm{ invertible, for } i = 1,2, \ldots, \ell \}.
\end{split}
\end{equation*}
By choosing $ \mathcal{C}_2 = \{ \mathbf{0} \} $, the same result holds for generalized sum-rank weights.
\end{theorem}
\begin{proof}
We only need to prove the first claim. The second claim follows from the first one and Proposition \ref{proposition as minimum rank weights}.

First, take invertible matrices $  A_i \in K_i^{n_i \times n_i} $, for $ i = 1,2, \ldots, \ell $, and define $ A = {\rm diag}(A_1, A_2, \ldots, A_\ell) \in \mathbb{F}^{n \times n} $. If $ \mathbf{d} = (\mathbf{d}^{(1)}, \mathbf{d}^{(2)}, \ldots, \mathbf{d}^{(\ell)}) \in \mathbb{F}^n $, where $ \mathbf{d}^{(i)} \in \mathbb{F}^{n_i} $, for $ i = 1,2, \ldots, \ell $, then we have that
\begin{equation*}
\begin{split}
 {\rm Supp}(\mathbf{d}A) & = ({\rm Row}(M_{\mathcal{A}_1}(\mathbf{d}^{(1)} A_1)), {\rm Row}(M_{\mathcal{A}_2}(\mathbf{d}^{(2)} A_2)), \ldots, {\rm Row}(M_{\mathcal{A}_\ell}(\mathbf{d}^{(\ell)} A_\ell)))  \\
  & = ({\rm Row}(M_{\mathcal{A}_1}(\mathbf{d}^{(1)})) A_1, {\rm Row}(M_{\mathcal{A}_2}(\mathbf{d}^{(2)})) A_2, \ldots, {\rm Row}(M_{\mathcal{A}_\ell}(\mathbf{d}^{(\ell)})) A_\ell).
\end{split}
\end{equation*}
Moreover, since $ \mathcal{L}_i A_i + \mathcal{L}^\prime_i A_i = (\mathcal{L}_i + \mathcal{L}^\prime_i) A_i $, for any two subspaces $ \mathcal{L}_i, \mathcal{L}^\prime_i \in \mathcal{P}(K_i^{n_i}) $, for $ i = 1,2, \ldots, \ell $, we conclude by definition that $ {\rm Supp}(\mathcal{D}A) = {\rm Supp}(\mathcal{D}) A $, where $ {\rm Supp}(\mathcal{D}) A $ is defined in the straightforward way. Therefore, we deduce that
$$ {\rm wt}_{SR}(\mathcal{D}) = {\rm Rk}({\rm Supp}(\mathcal{D})) = {\rm Rk}({\rm Supp}(\mathcal{D}A)) = {\rm wt}_{SR}(\mathcal{D}A) \leq {\rm wt}_{H}(\mathcal{D}A), $$
and the inequality $ \leq $ follows.

Second, by linear algebra there exist invertible matrices $ A_i \in K_i^{n_i \times n_i} $ such that 
$$ \mathcal{L}_i A_i = K_i^{r_i} \times \{ 0 \}^{n_i - r_i}, $$
where $ r_i = \dim_{K_i}(\mathcal{L}_i) $, for $ i = 1,2, \ldots, \ell $, and where $ {\rm Supp}(\mathcal{D}) = (\mathcal{L}_1, \mathcal{L}_2, \ldots, \mathcal{L}_\ell) $. If $ A = {\rm diag}(A_1, A_2, \ldots, A_\ell) \in \mathbb{F}^{n \times n} $, then we deduce that 
$$ \mathcal{V}_{\boldsymbol{\mathcal{L}}} A = \prod_{i=1}^\ell (\mathbb{F}^{r_i} \times \{ 0 \}^{n_i - r_i}), $$
where the latter is the Hamming support space of $ \mathcal{D} A $. Thus $ {\rm wt}_{SR}(\mathcal{D}) = {\rm wt}_{SR}(\mathcal{D}A) = {\rm wt}_H(\mathcal{D} A) $. Hence the inequality $ \geq $ follows.
\end{proof}

\begin{remark}
As observed in \cite[Sec. II]{tsfasman} for $ \mathbf{n} = \mathbf{1} $ and in \cite[Th. 5]{similarities} for $ \ell = 1 $, we directly deduce from the previous theorem that sum-rank weights of subspaces and relative generalized sum-rank weights are invariant by changes of bases (Definition \ref{def change of bases}). This is key in describing generalized Hamming weights by projective systems \cite{tsfasman}.
\end{remark}

We automatically deduce the following result.

\begin{corollary} \label{cor d_SR for finer partitions}
Assume that $ n_i = \sum_{j=1}^{v_i} n_{i,j} $, for $ i = 1,2, \ldots, \ell $. Denote by $ {\rm d}_{SR}^{ref} $ the sum-rank metric with respect to the refined partition $ n = \sum_{i=1}^\ell \sum_{j=1}^{v_i} n_{i,j} $. If $ \mathcal{C}_2 \subsetneqq \mathcal{C}_1 \subseteq \mathbb{F}^n $ are linear codes, then for all $ r = 1,2, \ldots, \dim(\mathcal{C}_1 / \mathcal{C}_2) $, it holds that
\begin{equation*}
\begin{split}
{\rm d}_{SR,r}(\mathcal{C}_1, \mathcal{C}_2) = \min \{ & {\rm d}_{SR,r}^{ref}(\mathcal{C}_1 A, \mathcal{C}_2 A) \mid A = {\rm diag}(A_1, A_2, \ldots, A_\ell), \\
 & A_i \in K_i^{n_i \times n_i} \textrm{ invertible, for } i = 1,2, \ldots, \ell \}.
\end{split}
\end{equation*}
By choosing $ \mathcal{C}_2 = \{ \mathbf{0} \} $, the same result holds for generalized sum-rank weights.
\end{corollary}

Hence the following result extends \cite[Th. 7]{similarities} from connecting the extremal cases $ \mathbf{n} = \mathbf{1} $ and $ \ell = 1 $ to connecting all intermediate cases.

\begin{theorem} \label{th comparison bounds}
Fix $ k $, and choose positive integers $ r,s = 1,2, \ldots, k $ and functions $ f_{r,s}, g_{r,s} : \mathbb{N} \longrightarrow \mathbb{R} $, which may also depend on $ \mathbf{n}, m, k $ and the sizes of $ K_1, K_2, \ldots, K_\ell $. If $ g_{r,s} $ is non-decreasing, then every bound of the form 
$$ f_{r,s}(d_r (\mathcal{C}_1, \mathcal{C}_2)) \geq g_{r,s}(d_s (\mathcal{C}_1, \mathcal{C}_2)) $$
that is valid for Hamming weights $ {\rm d}_H $ (or refined sum-rank weights $ {\rm d}_{SR}^{ref} $ as in the previous corollary), for any nested linear code pair $ \mathcal{C}_2 \subsetneqq \mathcal{C}_1 \subseteq \mathbb{F}^n $ with $ \dim(\mathcal{C}_1 / \mathcal{C}_2) = k $, is also valid for sum-rank weights $ {\rm d}_{SR} $.  
\end{theorem}
\begin{proof}
Thanks to Theorem \ref{th connection GHW and GSRW}, the proof can be translated mutatis mutandis from that in \cite[Th. 7]{similarities}.
\end{proof}

Just as in \cite[Sec. V]{similarities}, we may apply the bounds from \cite{kloeve} and \cite[P. I, Subsec. III-A]{tsfasman} to relative generalized sum-rank weights. In the case where $ q = |\mathbb{F}| < \infty $, we may give the following bounds, among many others, where $ 1 \leq r < s \leq k $:
\begin{enumerate}
\item
Monotonicity and its refinement (see \cite[Th. 1]{kloeve} and \cite[Eq. (18)]{tsfasman}): 
$$ d_{r+1} \geq d_{r} + 1, $$
$$ (q^{s} - q^{s-r})d_s \geq (q^{s}-1)d_r. $$
\item 
Generalized Griesmer bound (see \cite[Eq. (14)]{tsfasman} and \cite[Eq. (16)]{tsfasman}): 
$$ d_s \geq d_r + \sum_{i=0}^{s-r} \left\lceil \frac{(q-1)d_r}{(q^{r}-1)q^{i}} \right\rceil. $$
\item 
Another bound (see \cite[Eq. (20)]{tsfasman}): 
$$ d_r \geq n - \left\lfloor \frac{(q^{k-r}-1)(n-d_s)}{q^{k-s}-1} \right\rfloor. $$
\end{enumerate}

Without Theorem \ref{th connection GHW and GSRW}, even in the case $ \ell = 1 $, proving the monotonicity bound and its refinement from scratch requires proofs that are not so short. See \cite[Prop. II.3]{jerome}, \cite[Sec. II]{rgrw} or \cite[Sec. IV]{oggier}, for instance.

Also from Theorem \ref{th connection GHW and GSRW} we deduce that any bound of the form $ d_r \leq M $, for a fixed number $ M > 0 $, that is valid for sum-rank weights is valid for less refined sum-rank weights. In particular, each of the previous bounds has a corresponding version of this form: Monotonicity gives the Singleton bound (see next subsection), its refinement gives the Plotkin bound (\cite[Th. 2]{kloeve} or \cite[Eq. (9)]{tsfasman}), and the previous Griesmer bound gives the classical form of the Griesmer bound (\cite[Th. 4]{kloeve}). We may similarly obtain asymptotic upper bounds as in \cite[P. I, Subsec. V-B]{tsfasman}, which we leave to the reader. Existential bounds however seem harder to obtain. We leave them as open problem.

\subsection{Maximum sum-rank distance codes} \label{subsec MSRD codes}

In this subsection, we derive a Singleton bound on relative generalized sum-rank weights. We extend the notion of maximum sum-rank distance (MSRD) codes from \cite[Subsec. 3.3]{linearizedRS} and define MSRD ranks. We then connect the MSRD rank of a code with the minimum sum-rank distance of its dual. Thanks to it, we conclude that the dual of an MSRD code is again MSRD, which has not been proven yet. Finally, we characterize MSRD codes in terms of sum-rank supports and prove that any restriction or shortening of an MSRD code is in turn MSRD. All results in this subsection can be stated in terms of the parameter $ {\rm K}_{SR, \mu}(\mathcal{C}_1, \mathcal{C}_2) $. The formulas would be exactly as in \cite{forney, luo} and are left to the reader.

We start by stating the Singleton bound, which follows directly from \cite[Eq. (24)]{luo}.

\begin{proposition}[\textbf{Generalized Singleton bound}] \label{prop singleton bound}
Let $ \mathcal{C}_2 \subsetneqq \mathcal{C}_1 \subseteq \mathbb{F}^n $ be nested linear codes with $ k_1 = \dim(\mathcal{C}_1) $ and $ k = \dim(\mathcal{C}_1 / \mathcal{C}_2) $. It holds that
$$ {\rm d}_{SR,r}(\mathcal{C}_1, \mathcal{C}_2) \leq n - k_1 + r, $$
for $ r = 1,2, \ldots, k $. 
\end{proposition}

The case $ \mathcal{C}_2 = \{ \mathbf{0} \} $ gives \cite[Cor. 1]{wei} when $ \mathbf{n} = \mathbf{1} $, and then choosing $ r = 1 $ gives the classical Singleton bound \cite{singleton}. As in \cite{linearizedRS}, the case $ r=1 $ for one linear code (i.e. $ \mathcal{C}_2 = \{ \mathbf{0} \} $) motivates the following definition.

\begin{definition}[\textbf{MSRD codes \cite{linearizedRS}}]
A linear code $ \mathcal{C} \subseteq \mathbb{F}^n $ is maximum sum-rank distance (MSRD) if 
$$ d_{SR}(\mathcal{C}) = n - \dim(\mathcal{C}) + 1. $$
\end{definition}

By Corollary \ref{cor d_SR for finer partitions}, MRD codes, such as Gabidulin codes \cite{gabidulin, roth}, are also MSRD (for any length partition). However, their field size is always exponential in the code length. Linearized Reed-Solomon codes \cite[Def. 31]{linearizedRS} constitute the first and only known family of MSRD codes with field sizes that are subexponential in the code length. Their (relative) generalized sum-rank weights are given as follows:

\begin{proposition}
Let $ \mathcal{C}_2 \subsetneqq \mathcal{C}_1 \subseteq \mathbb{F}^n $ be nested linear codes with $ k_1 = \dim(\mathcal{C}_1) $ and $ k = \dim(\mathcal{C}_1 / \mathcal{C}_2) $, and where $ \mathcal{C}_1 $ is MSRD (for instance, a linearized Reed-Solomon code \cite{linearizedRS}). Then it holds that
$$ {\rm d}_{SR,r}(\mathcal{C}_1, \mathcal{C}_2) = n - k_1 + r, $$
for $ r = 1,2, \ldots, k $. 
\end{proposition}

More generally, if a linear code or code pair achieves the Singleton bound for a given $ r $, it achieves it for all $ s \geq r $. This discussion motivates the definition of MSRD ranks, which recovers \cite[Def. 1]{jerome} in the case $ \ell = 1 $ (see \cite[Sec. VI]{wei} for the case $ \mathbf{n} = \mathbf{1} $).

\begin{definition} [\textbf{MSRD rank}]
Given a $ k $-dimensional linear code $ \mathcal{C} \subseteq \mathbb{F}^n $, we define its MSRD rank as the minimum integer $ r = 1,2, \ldots, k $ such that $ {\rm d}_{SR,r}(\mathcal{C}) = n - \dim(\mathcal{C}) + r $, if such an $ r $ exists. In such a case, we say that $ \mathcal{C} $ is an $ r $-MSRD code.
\end{definition}

\begin{remark}
Codes without an MSRD rank as in the previous definition are precisely sum-rank degenerate codes as in the next subsection. The restricted code in its sum-rank effective length does have an MSRD rank.
\end{remark}

\begin{remark}
Just as we have extended the notion of $ r $-MDS and $ r $-MRD codes to $ r $-MSRD codes, we may analogously extend the notion of almost MDS \cite{almostMDS} and near MDS \cite{nearMDS} codes from the case $ \mathbf{n} = \mathbf{1} $ to the general case. It is worth noting that, in the case $ \ell = 1 $, an essentially different concept called quasi MRD codes \cite[Def. 10]{quasiMRD} may be introduced. Properties and constructions of almost MSRD, near MSRD and quasi MSRD codes would be of interest and are left open. 
\end{remark}

Hence we deduce the following result from Theorem \ref{th GSRW duality}, recovering \cite[Cor. III.3]{jerome} in the case $ \ell = 1 $. The case $ \mathbf{n} = \mathbf{1} $ can be found in \cite[Prop. 4.1]{tsfasman}.

\begin{corollary} \label{cor MSRD ranks and duals}
Given a $ k $-dimensional linear code $ \mathcal{C} \subseteq \mathbb{F}^n $, its MSRD rank $ r $ satisfies that
$$ r = k - {\rm d}_{SR}(\mathcal{C}^\perp) + 2. $$
\end{corollary}

We may now conclude that the dual of a linear MSRD code is again MSRD. This result recovers \cite[Th. 3]{gabidulin} in the case $ \ell = 1 $, and the well-known result in the case $ \mathbf{n} = \mathbf{1} $. It follows directly from the definitions and Corollary \ref{cor MSRD ranks and duals}.

\begin{theorem} \label{th dual of MSRD is MSRD}
A linear code $ \mathcal{C} \subseteq \mathbb{F}^n $ is MSRD if, and only if, its dual $ \mathcal{C}^\perp \subseteq \mathbb{F}^n $ is MSRD.
\end{theorem}

This result is new. As a particular case, it was shown in \cite[Th. 4]{secure-multishot} that the dual of a linearized Reed-Solomon code \cite[Def. 31]{linearizedRS} over finite fields is in turn a linearized Reed-Solomon code, hence MSRD. 

As it was the case with bounds, the $ r $-MSRD condition is related between different refinements of sum-rank metrics. The following result extends \cite[Prop. 6]{similarities} from connecting the cases $ \mathbf{n} = \mathbf{1} $ and $ \ell = 1 $ to connecting all cases. The case $ r=1 $ for general sum-rank metrics was given in \cite[Cor. 2]{universal-lrc}.

\begin{proposition}
Given $ k $ and $ 1 \leq r \leq k $, a $ k $-dimensional linear code $ \mathcal{C} \subseteq \mathbb{F}^n $ is $ r $-MSRD if, and only if, $ \mathcal{C} A \subseteq \mathbb{F}^n $ is $ r $-MDS (or $ r $-MSRD for a refined sum-rank metric $ {\rm d}_{SR}^{ref} $ as in Corollary \ref{cor d_SR for finer partitions}), for all $ A = {\rm diag}(A_1, A_2, \ldots, A_\ell) \in \mathbb{F}^{n \times n} $, such that $ A_i \in K_i^{n_i \times n_i} $ is invertible, for $ i = 1,2, \ldots, \ell $.
\end{proposition}
\begin{proof}
The result follows directly from Theorem \ref{th connection GHW and GSRW}.
\end{proof}

\begin{remark}
Setting $ r = 1 $ in the previous proposition, one may easily derive characterizations of linear MSRD codes, based on their generator matrices, from characterizations of MDS codes. When $ \ell = 1 $, the results in \cite[Sec. 3]{new-criteria} are recovered as particular cases.
\end{remark}

Finally, as in \cite[Props. 7 \& 8]{similarities} for both cases $ \ell =1 $ and $ \mathbf{n} = \mathbf{1} $, we may give a refinement of Corollary \ref{cor MSRD ranks and duals} in terms of pre-shortened codes (or restricted or shortened codes by Proposition \ref{prop dimensions restricted shortened}). Observe that this gives an improvement on the estimates of the dimensions of restricted codes in Corollary \ref{cor parameters restricted and shortened}. More importantly, they constitute characterizations of $ r $-MSRD codes in terms of sum-rank supports.

\begin{proposition} \label{prop MSRD rank for intermediate info leakage}
Given a $ k $-dimensional linear code $ \mathcal{C} \subseteq \mathbb{F}^n $ and an integer $ r = 1, $ $2, $ $ \ldots, k $, the following are equivalent:
\begin{enumerate}
\item
$ {\rm d}_{SR,r}(\mathcal{C}^\perp) = k + r $.
\item
$ {\rm d}_{SR}(\mathcal{C}) > n - k - r + 1 $.
\item
For all $ \boldsymbol{\mathcal{L}} \in \mathcal{P}(\mathbf{K}^\mathbf{n}) $ such that $ {\rm Rk}(\boldsymbol{\mathcal{L}}) \leq n- k - r + 1 $, we have that $ \mathcal{C} \cap \mathcal{V}_{\boldsymbol{\mathcal{L}}} = \{ \mathbf{0} \} $. Equivalently, $ \mathcal{C}^{\boldsymbol{\mathcal{L}}} = \{ \mathbf{0} \} $ or $ \dim((\mathcal{C}^\perp)_{\boldsymbol{\mathcal{L}}}) = {\rm Rk}(\boldsymbol{\mathcal{L}}) $.
\item
For all $ \boldsymbol{\mathcal{L}} \in \mathcal{P}(\mathbf{K}^\mathbf{n}) $ such that $ {\rm Rk}(\boldsymbol{\mathcal{L}}) \geq k + r - 1 $, we have that $ \mathcal{C} \cap \mathcal{V}_{\boldsymbol{\mathcal{L}}}^\perp = \{ \mathbf{0} \} $. Equivalently, $ \mathcal{C}^{(\boldsymbol{\mathcal{L}}^\perp)} = \{ \mathbf{0} \} $ or $ \dim(\mathcal{C}_{\boldsymbol{\mathcal{L}}}) = \dim(\mathcal{C}) $.
\end{enumerate}
\end{proposition}
\begin{proof}
The equivalence between Items 1 and 2 is Corollary \ref{cor MSRD ranks and duals}, and the equivalence between Items 2, 3 and 4 is trivial from the definitions using pre-shortened codes. The equivalent statements using restricted and shortened codes follow from Proposition \ref{prop dimensions restricted shortened}. 
\end{proof}

Combining the previous proposition with Corollary \ref{cor parameters restricted and shortened}, we obtain the following:

\begin{corollary}
Any restriction or shortening of a linear MSRD code gives a linear MSRD code.
\end{corollary}

We recall that, from \cite[Subsec. V-F]{secure-multishot}, any restriction or shortening of a linearized Reed-Solomon code is in turn a linearized Reed-Solomon code.

\subsection{Sum-rank effective length and degenerate codes} \label{subsec sum-rank degenerate}

In this subsection, we define the sum-rank effective length and degeneratess of linear codes (see \cite{forney} for the case $ \mathbf{n} = \mathbf{1} $). We characterize them in terms of the dual code, as done in \cite[Sec. 6]{slides} when $ \ell = 1 $. However, we follow the approach used in \cite[Def. 9]{similarities} for $ \ell = 1 $, which is intrinsic to the corresponding metric. 

\begin{definition}[\textbf{Sum-rank effective length}]
Given a linear code $ \mathcal{C} \subseteq \mathbb{F}^n $, we define its sum-rank effective length as the minimum integer $ N = 1,2, \ldots, n $ such that $ N = N_1 + N_2 + \cdots + N_\ell $, with $ 0 \leq N_i \leq n_i $ for $ i = 1,2, \ldots, \ell $, and such that there exists a linear sum-rank isometry $ \phi : \mathbb{F}^N \longrightarrow \mathbb{F}^n $ satisfying that $ \mathcal{C} \subseteq \phi(\mathbb{F}^N) $, where the sum-rank metric in $ \mathbb{F}^N $ corresponds to the previous partition of $ N $.
\end{definition}

\begin{definition}[\textbf{Sum-rank degenerate codes}] \label{def degenerate}
Given a linear code $ \mathcal{C} \subseteq \mathbb{F}^n $, we say that it is sum-rank degenerate if its sum-rank effective length $ N $ satisfies that $ N < n $.
\end{definition}

Observe that isometries are one to one by the fact that $ {\rm d}(\mathbf{c}, \mathbf{d}) = 0 $ if, and only if, $ \mathbf{c} = \mathbf{d} $, for any metric $ {\rm d} $. Thus the previous definition means that we may consider $ \mathcal{C} $ in $ \mathbb{F}^N $, for a strictly smaller length $ N < n $.

Thanks to the characterization of sum-rank support spaces from Item 5 in Theorem \ref{th charact sum-rank support spaces}, we may now easily connect the sum-rank effective length of a linear code with its last generalized sum-rank weight. This result extends \cite[Prop. 3]{similarities} when $ \ell = 1 $. The case $ \mathbf{n} = \mathbf{1} $ is trivial.

\begin{proposition}
Given a linear code $ \mathcal{C} \subseteq \mathbb{F}^n $, its sum-rank effective length is 
$$ N = {\rm Rk}({\rm Supp}(\mathcal{C})) = {\rm wt}_{SR}(\mathcal{C}) = {\rm d}_{SR,k}(\mathcal{C}) = $$
$$ n - \max \{ r \in [k] \mid {\rm d}_{SR,r}(\mathcal{C}^\perp) = r \}, $$ 
where $ k = \dim(\mathcal{C}) $, and the last maximum is defined as $ 0 $ if the set is empty. 
\end{proposition}
\begin{proof}
Let $ N $ be the sum-rank effective length of $ \mathcal{C} $ and let $ \phi : \mathbb{F}^N \longrightarrow \mathbb{F}^n $ be a linear sum-rank isometry as in Definition \ref{def degenerate}. Define $ \mathcal{V} = \phi(\mathbb{F}^N) \subseteq \mathbb{F}^n $. By Item 5 in Theorem \ref{th charact sum-rank support spaces}, there exists $ \boldsymbol{\mathcal{L}} \in \mathcal{P}(\mathbf{K}^\mathbf{n}) $ such that $ \mathcal{V} = \mathcal{V}_{\boldsymbol{\mathcal{L}}} $. Since $ \mathcal{C} \subseteq \phi(\mathbb{F}^N) = \mathcal{V}_{\boldsymbol{\mathcal{L}}} $, we deduce from Corollary \ref{cor min support space containing} that
$$ {\rm d}_{SR,k}(\mathcal{C}) = {\rm Rk}({\rm Supp}(\mathcal{C})) \leq {\rm Rk}(\boldsymbol{\mathcal{L}}) = \dim(\mathcal{V}) = N. $$

Conversely, let $ \boldsymbol{\mathcal{L}} = {\rm Supp}(\mathcal{C}) $. Again, by Item 5 in Theorem \ref{th charact sum-rank support spaces}, there exists a bijective linear sum-rank isometry $ \phi : \mathbb{F}^{N^\prime} \longrightarrow \mathcal{V}_{\boldsymbol{\mathcal{L}}} $, for $ N^\prime = \dim(\mathcal{V}_{\boldsymbol{\mathcal{L}}}) = {\rm d}_{SR,k}(\mathcal{C}) $. Since $ \mathcal{C} \subseteq \mathcal{V}_{\boldsymbol{\mathcal{L}}} $, we conclude that $ N \leq N^\prime = {\rm d}_{SR,k}(\mathcal{C}) $ by definition of $ N $.

The last equality follows directly from Theorem \ref{th GSRW duality}.
\end{proof}

We may now characterize sum-rank degenerate linear codes in terms of the minimum sum-rank distance of the dual code, as done in \cite[Sec. 6]{slides} when $ \ell = 1 $. 

\begin{corollary} \label{cor degenerate distance dual}
Given a linear code $ \mathcal{C} \subseteq \mathbb{F}^n $, the following are equivalent.
\begin{enumerate}
\item
$ \mathcal{C} $ is sum-rank degenerate.
\item
$ {\rm d}_{SR,k}(\mathcal{C}) < n $ or, equivalently, $ {\rm Supp}(\mathcal{C}) \neq \boldsymbol{\mathcal{T}} = (K_1^{n_1}, K_2^{n_2}, \ldots, K_\ell^{n_\ell}) $.
\item
$ {\rm d}_{SR}(\mathcal{C}^\perp) = 1 $.
\end{enumerate}
In particular, linear MSRD codes are never sum-rank degenerate.
\end{corollary}

As shown in \cite[Cor. 6.5]{slides} for the case $ \ell = 1 $, there are certain choices of parameters for which linear codes are always sum-rank degenerate. We now extend that result to the general case.

\begin{proposition}
Let $ \mathcal{C} \subseteq \mathbb{F}^n $ be a $ k $-dimensional linear code with $ k_i = \dim(\pi_i(\mathcal{C})) $, where $ \pi_i : \mathbb{F}^n \longrightarrow \mathbb{F}^{n_i} $ denotes the projection onto the $ i $th block of coordinates, for $ i = 1,2, \ldots, \ell $. It holds that
$$ {\rm d}_{SR,k}(\mathcal{C}) \leq \sum_{i=1}^\ell k_i m_i . $$
In particular, if $ \sum_{i=1}^\ell k_i m_i < n $, then $ \mathcal{C} $ is sum-rank degenerate.
\end{proposition}
\begin{proof}
It follows by combining Item 2 in Theorem \ref{th charact sum-rank support spaces}, Corollary \ref{cor degenerate distance dual} and the fact that $ {\rm wt}_R(\pi_i(\mathcal{C})) \leq k_i m_i $ (see the proof of \cite[Cor. 6.5]{slides}), for $ i = 1,2, \ldots, \ell $.
\end{proof}

Observe that, in the case $ \mathbf{n} = \mathbf{1} $ and $ m_1 = m_2 = \ldots = m_\ell = 1 $, the previous proposition is trivially equivalent to the definition of Hamming-metric degenerateness.

\section*{Acknowledgement}

The author wishes to thank Frank R. Kschischang for valuable discussions on this manuscript. The author also whishes to thank the anonymous reviewers for the valuable comments on this work. The author also gratefully acknowledges the support from The Independent Research Fund Denmark (Grant No. DFF-7027-00053B).

\footnotesize
 

\normalsize

\appendix

\section{The skew metric and skew supports} \label{app}

In this appendix, we briefly revisit the relation between the sum-rank metric (Definition \ref{def sum-rank metric}) and the \textit{skew metric} introduced in \cite[Def. 9]{linearizedRS}. We extend such a relation to sum-rank supports and \textit{skew supports} (which we introduce in this appendix), and the corresponding support spaces. The exposition in this appendix follows the lines in \cite{linearizedRS}.
 
Let $ \sigma : \mathbb{F} \longrightarrow \mathbb{F} $ be a field endomorphism and let $ \delta : \mathbb{F} \longrightarrow \mathbb{F} $ be a $ \sigma $-derivation, that is, $ \delta $ is additive and $ \delta(ab) = \sigma(a) \delta(b) + \delta(a) b $, for all $ a,b \in \mathbb{F} $. Define the skew polynomial ring $ \mathbb{F}[x; \sigma, \delta] $ as the vector space over $ \mathbb{F} $ with basis $ \{ x^i \mid i \in \mathbb{N} \} $ and with product given by the rules $ x^i x^j = x^{i + j} $, for $ i,j \in \mathbb{N} $, and
\begin{equation}
xa = \sigma(a) x + \delta(a),
\label{eq product commutativity rule}
\end{equation}
for $ a \in \mathbb{F} $. Define the degree of a non-zero skew polynomial $ F = \sum_{i \in \mathbb{N}} F_i x^i \in \mathbb{F}[x; \sigma, \delta] $, denoted by $ \deg(F) $, as the maximum $ i \in \mathbb{N} $ such that $ F_i \neq 0 $. We also define $ \deg (0) = \infty $. Skew polynomial rings were introduced by Ore in \cite{ore} and the products given by (\ref{eq product commutativity rule}) are the only products in $ \mathbb{F}[x; \sigma, \delta] $ such that $ \deg(FG) = \deg(F) + \deg(G) $, for $ F,G \in \mathbb{F}[x; \sigma, \delta] $. The extension to several variables was recently given in \cite{multivariateskew}. Conventional polynomial rings are recovered by setting $ \sigma = {\rm Id} $ and $ \delta = 0 $.

Since $ \mathbb{F}[x ; \sigma, \delta] $ is a right Euclidean domain, we may define the \textit{evaluation} of $ F \in \mathbb{F}[x; \sigma, \delta] $ in $ a \in \mathbb{F} $ as the unique $ F(a) \in \mathbb{F} $ such that there exists $ G \in \mathbb{F}[x; \sigma, \delta] $ with
$$ F = G \cdot (x-a) + F(a). $$
This concept of evaluation was introduced in \cite{lam, lam-leroy}.

Given a subset $ \Omega \subseteq \mathbb{F} $, we may define its associated ideal as $ I(\Omega) = \{ F \in \mathbb{F}[x;\sigma, \delta] \mid F(a) = 0, \forall a \in \Omega \} $. Observe that $ I(\Omega) $ is a left ideal in $ \mathbb{F}[x;\sigma, \delta] $. Since $ \mathbb{F}[x;\sigma, \delta] $ is a right Euclidean domain, there exists a unique monic skew polynomial $ F_\Omega \in I(\Omega) $ of minimal degree among those in $ I(\Omega) $, which in turn generates $ I(\Omega) $ as left ideal. Such a skew polynomial is called the \textit{minimal skew polynomial} of $ \Omega $ \cite{lam-leroy}. 

Next, given a subset $ \Omega \subseteq \mathbb{F} $, we define its \textit{P-closure} as $ \overline{\Omega} = Z(F_\Omega) \subseteq \mathbb{F} $ (the \textit{set of zeros} of $ F_\Omega $), and we say that $ \Omega $ is \textit{P-closed} if $ \overline{\Omega} = \Omega $. A set $ \Omega \subseteq \mathbb{F} $ is called \textit{P-independent} if $ a \notin \overline{\Omega \setminus \{ a \}} $, for all $ a \in \Omega $. We say that $ \mathcal{B} \subseteq \Omega $ is a \textit{P-basis} of a P-closed set $ \Omega $ if $ \mathcal{B} $ is P-independent and $ \Omega = \overline{\mathcal{B}} $. We also say that $ \Omega $ is a finitely generated P-closed set if it admits a finite P-basis, which is the case as long as $ \Omega \neq \mathbb{F} $, or $ \Omega = \mathbb{F} $ and $ \mathbb{F} $ is finite. 

Given a finitely generated P-closed set $ \Omega \subseteq \mathbb{F} $, any two of its P-bases are finite and have the same number of elements, which moreover coincides with $ \deg(F_\Omega) $. This motivates the definition \textit{rank} of $ \Omega $ as
$$ {\rm Rk}(\Omega) = \deg(F_\Omega) < \infty. $$

Fix a finitely generated P-closed set $ \Omega \subseteq \mathbb{F} $ of rank $ n $ and fix one of its P-bases $ \mathcal{B} $. An important tool to define skew metrics is skew polynomial Lagrange interpolation. Let $ \mathbb{F}[x; \sigma, \delta]_n $ be the $ n $-dimensional vector space of skew polynomials of degree less than $ n $. It follows from \cite[Th. 8]{lam} that the \textit{evaluation map} over the points in $ \mathcal{B} $,
$$ E_\mathcal{B} : \mathbb{F}[x; \sigma, \delta]_n \longrightarrow \mathbb{F}^\mathcal{B}, $$
is a vector space isomorphism. Hence we may define \textit{skew weights} \cite[Def.~9]{linearizedRS} as follows.

\begin{definition} [\textbf{Skew weights \cite{linearizedRS}}] \label{def skew weights}
Given $ F \in \mathbb{F}[x; \sigma, \delta]_n $ and $ f = E_\mathcal{B}(F) \in \mathbb{F}^\mathcal{B} $, we define their skew weight over $ \Omega $ as
$$ {\rm wt}_\mathcal{B}(f) = {\rm wt}_\Omega(F) = n - {\rm Rk}(Z_\Omega(F)), $$
where $ Z_\Omega(F) = Z(F) \cap \Omega = Z(\{ F, F_\Omega \}) $ is the P-closed set of zeros of $ F $ in $ \Omega $. 
\end{definition}

Skew weights are indeed weights \cite[Prop. 10]{linearizedRS} and define a metric in $ \mathbb{F}^\mathcal{B} $, called the \textit{skew metric} \cite[Def. 11]{linearizedRS}, by the usual formula: $ {\rm d}_\mathcal{B}(f, g) = {\rm wt}_\mathcal{B}(f - g) $, for $ f,g \in \mathbb{F}^\mathcal{B} $. To relate this metric with the sum-rank metric, we need the concept of conjugacy from \cite{lam-leroy}: We say that $ a, c \in \mathbb{F} $ are \textit{conjugates} if there exists $ \beta \in \mathbb{F}^* $ such that 
$$ c = a^\beta \stackrel{def}{=} \sigma(\beta)\beta^{-1} a + \delta(\beta)\beta^{-1}. $$ 
Putting together the results \cite[Th. 23]{lam} and \cite[Th. 4.5]{lam-leroy}, we obtain the following characterization: A finite subset $ \mathcal{B} \subseteq \mathbb{F} $ with $ n $ elements is a P-basis of $ \Omega = \overline{\mathcal{B}} $ if, and only if, $ n = n_1 + n_2 + \cdots + n_\ell $, for some $ \ell $, and there exists pair-wise non-conjugate elements $ a^{(1)}, a^{(2)}, \ldots, a^{(\ell)} \in \mathbb{F} $ and a set of linearly independent elements $ \{ \beta_1^{(i)}, \beta_2^{(i)}, \ldots, \beta_{n_i}^{(i)} \} \subseteq \mathbb{F} $, over the subfield $ K_i = K_{a^{(i)}} = \{ \beta \in \mathbb{F}^* \mid \left( a^{(i)} \right)^{\beta} = a^{(i)} \} \cup \{ 0 \} \subseteq \mathbb{F} $, for each $ i = 1,2, \ldots, \ell $, such that
\begin{equation}
\mathcal{B} = \bigcup_{i = 1}^\ell \left\lbrace \left( a^{(i)} \right)^{\beta_j^{(i)}} \mid j = 1,2, \ldots, n_i \right\rbrace ,
\label{eq charact of P-basis}
\end{equation}
where the union is disjoint. With this characterization at hand, we may give a vector space isomorphism connecting both metrics. The result follows from \cite[Th. 2 \& 3]{linearizedRS}.

\begin{theorem} [\textbf{\cite{linearizedRS}}]
With notation as above, define the vector space isomorphism $ \phi_\mathcal{B} : \mathbb{F}^n \longrightarrow \mathbb{F}^\mathcal{B} $ by $ \phi_\mathcal{B} (\mathbf{c}^{(1)}, $ $ \mathbf{c}^{(2)}, $ $ \ldots, $ $ \mathbf{c}^{(\ell)}) = f $, where $ \mathbf{c}^{(i)} = (c_1^{(i)}, c_2^{(i)}, \ldots, c_{n_i}^{(i)}) \in \mathbb{F}^{n_i} $ and
\begin{equation}
f \left( \left( a^{(i)} \right)^{\beta_j^{(i)}} \right) = c_j^{(i)} (\beta_j^{(i)})^{-1},
\label{eq transforming the received word}
\end{equation}
for $ j = 1,2, \ldots, n_i $ and $ i = 1,2,\ldots, \ell $. Then $ \phi_\mathcal{B} $ is an isometry: For $ \mathbf{c} \in \mathbb{F}^n $, it holds that
$$ {\rm wt}_\mathcal{B}(\phi_\mathcal{B}(\mathbf{c})) = {\rm wt}_{SR}(\mathbf{c}), $$
where $ {\rm wt}_{SR} $ is the sum-rank weight from Definition \ref{def sum-rank metric} with $ K_i = K_{a^{(i)}} $, for $ i = 1,2, \ldots, \ell $.
\end{theorem}

The representation (\ref{eq charact of P-basis}) and the map given by (\ref{eq transforming the received word}) establish a dictionary between the sum-rank metric and the skew metric. This dictionary, however, depends on the conjugacy representatives $ a^{(1)}, a^{(2)}, \ldots, a^{(\ell)} $ and the P-basis $ \mathcal{B} $ of $ \Omega $. The elements $ \beta_1^{(i)}, \beta_2^{(i)}, $ $ \ldots, $ $ \beta_{n_i}^{(i)} $ are determined up to scalar factor in $ K_i^* $ (thus uniquely as projective points in $ \mathbb{P}_{K_i}(\mathbb{F}) $) by the conjugacy representatives and $ \mathcal{B} $, for $ i = 1,2, \ldots, \ell $. It is important to notice that in the case $ \sigma = {\rm Id} $ and $ \delta = 0 $, which corresponds to conventional polynomials and the Hamming metric, the dependency disappears since conjugacy classes only have one element and the only P-basis of $ \Omega $ is $ \mathcal{B} = \Omega $.

In particular, the concept of sum-rank support can be readily translated into the concept of \textit{skew support}. First, define the \textit{lattice of skew supports} in $ \Omega $ as
$$ \mathcal{P}_{Sk}(\Omega) = \{ \Psi \subseteq \Omega \mid \Psi \textrm{ is P-closed} \}. $$
Thus skew supports will simply be P-closed subsets of $ \Omega $, which form a lattice with intersections $ \Psi_1 \cap \Psi_2 $ and sums defined as $ \Psi_1 + \Psi_2 = \overline{\Psi_1 \cup \Psi_2} = Z(F_{\Psi_1 \cup \Psi_2}) $. The results \cite[Prop. 43]{linearizedRS} and \cite[Prop. 47]{linearizedRS} state that $ \mathcal{P}_{Sk}(\Omega) $ is a lattice isomorphic to $ \mathcal{P}(\mathbf{K}^\mathbf{n}) $, by mapping P-bases into lists of bases via (\ref{eq charact of P-basis}), where $ K_i = K_{a^{(i)}} $, for $ i = 1,2, \ldots, \ell $. This mapping will now be used to define skew supports. As for vector and projective spaces, we implicitly associate the zero vector space with the empty P-closed set.

\begin{definition}[\textbf{Skew supports}] \label{def skew support}
With notation as above, let $ f \in \mathbb{F}^\mathcal{B} $ and define $ \mathbf{c} = (\mathbf{c}^{(1)}, \mathbf{c}^{(2)}, \ldots, \mathbf{c}^{(\ell)}) = \phi_\mathcal{B}^{-1}(f) \in \mathbb{F}^n $, where $ \mathbf{c}^{(i)} \in \mathbb{F}^{n_i} $, for $ i = 1,2, \ldots, \ell $. Next, let $ \gamma_h^{(i)} = \sum_{j=1}^{n_i} c_{h,j}^{(i)} \beta_j^{(i)} \in \mathbb{F} $, where $ (c_{h,1}^{(i)}, c_{h,2}^{(i)}, \ldots, c_{h,n_i}^{(i)}) \in K_i^{n_i} $ form the rows of $ M_{\mathcal{A}_i}(\mathbf{c}^{(i)}) \in K_i^{m_i \times n_i} $, for $ h = 1,2, \ldots, m_i $, and let $ \mathcal{G}_i \subseteq \mathbb{F}^* $ be a basis of the vector space generated by $ \gamma_1^{(i)}, \gamma_2^{(i)}, $ $ \ldots, $ $ \gamma_{m_i}^{(i)} \subseteq \mathbb{F} $ over $ K_i $, for $ i = 1,2, \ldots, \ell $. Define the P-independent set
$$ \mathcal{B}_f = \bigcup_{i = 1}^\ell \left\lbrace \left( a^{(i)} \right)^{\gamma} \mid \gamma \in \mathcal{G}_i \right\rbrace. $$
We define the skew support of $ f \in \mathbb{F}^\mathcal{B} $ as
$$ {\rm Supp}_{Sk}(f) = \Omega_f = \overline{\mathcal{B}_f} \in \mathcal{P}_{Sk}(\Omega). $$
Finally, for a vector subspace $ \mathcal{F} \subseteq \mathbb{F}^\mathcal{B} $, we define its skew support as
$$ {\rm Supp}_{Sk}(\mathcal{F}) = \sum_{f \in \mathcal{F}} {\rm Supp}_{Sk}(f) \in \mathcal{P}_{Sk}(\Omega), $$
which allows to define the skew weight of $ \mathcal{F} $ as $ {\rm wt}_\mathcal{B}(\mathcal{F}) = {\rm Rk}({\rm Supp}_{Sk}(\mathcal{F})) $.
\end{definition}

As it was the case with the map in (\ref{eq transforming the received word}), the skew support $ {\rm Supp}_{Sk}(f) \in \mathcal{P}_{Sk}(\Omega) $ depends only on the conjugacy representatives and the choice of P-basis $ \mathcal{B} $ of $ \Omega $. To see this, note that the vector space generated by the rows of $ M_{\mathcal{A}_i}(\mathbf{c}^{(i)}) $ does not depend on $ \mathcal{A}_i $, and secondly, the P-basis corresponding to different bases of the subspace generated by $ \gamma_1^{(i)}, \gamma_2^{(i)}, \ldots, \gamma_{m_i}^{(i)} \in \mathbb{F} $ over $ K_i $ generate the same P-closed set $ \Omega_f $ by \cite[Cor. 27]{linearizedRS}.

Using the same arguments, we may prove the following properties:

\begin{proposition}
The following properties hold.
\begin{enumerate}
\item
For $ f \in \mathbb{F}^\mathcal{B} $ and $ a \in \mathbb{F}^* $, it holds that $ {\rm Supp}_{Sk}(af) = {\rm Supp}_{Sk}(\langle f \rangle) = {\rm Supp}_{Sk}(f) $ and
$$ {\rm Rk}({\rm Supp}_{Sk}(f)) = {\rm wt}_\mathcal{B}(f). $$
\item
$ \phi_\mathcal{B}^{-1}(f) $ and $ \phi_\mathcal{B}^{-1}(g) $ have the same sum-rank support if, and only if, $ f $ and $ g $ have the same skew support, for $ f,g \in \mathbb{F}^\mathcal{B} $. The same holds for subspaces of $ \mathbb{F}^\mathcal{B} $.
\item
If $ \mathcal{F} \subseteq \mathbb{F}^\mathcal{B} $ and $ \mathcal{D} = \phi_\mathcal{B}^{-1}(\mathcal{F}) \subseteq \mathbb{F}^n $ are subspaces, then
$$ {\rm wt}_\mathcal{B}(\mathcal{F}) = {\rm Rk}({\rm Supp}_{Sk}(\mathcal{F})) = {\rm Rk}({\rm Supp}(\mathcal{D})) = {\rm wt}_{SR}(\mathcal{D}) . $$
\end{enumerate}
\end{proposition}

The concept of \textit{skew support space} may also be considered. It may be introduced as a lattice of subspaces of $ \mathbb{F}^\mathcal{B} $. 

\begin{definition} [\textbf{Skew support spaces}]
Given a P-closed subset $ \Psi \subseteq \Omega $ (i.e. $ \Psi \in \mathcal{P}_{Sk}(\Omega) $), we define the skew support space associated to $ \Psi $ over $ \mathcal{B} $ as
$$ \mathcal{W}_\Psi = \{ f \in \mathbb{F}^\mathcal{B} \mid {\rm Supp}_{Sk}(f) \subseteq \Psi \} \subseteq \mathbb{F}^\mathcal{B} . $$
\end{definition}

We may add to Theorem \ref{th charact sum-rank support spaces} the following characterizations. They follow from the results in this appendix, except for the arithmetic characterizations in Items 3 and 4. These follow by combining Item 6 in Theorem \ref{th charact sum-rank support spaces} and the recent result \cite[Th. 2]{pirlrc}, which gives the connection between coordinate-wise matrix products as in Theorem \ref{th charact sum-rank support spaces} and products of skew polynomials given by (\ref{eq product commutativity rule}).

\begin{proposition}
The following are equivalent:
\begin{enumerate}
\item
$ \mathcal{W} $ is a skew support space, that is, there exists $ \Psi \in \mathcal{P}_{Sk}(\Omega) $ such that $ \mathcal{W} = \mathcal{W}_\Psi $.
\item
$ \mathcal{V} = \phi_\mathcal{B}^{-1}(\mathcal{W}) \subseteq \mathbb{F}^n $ is a sum-rank support space.
\item
$ \mathcal{W} $ is a left ideal of $ \mathbb{F}^\mathcal{B} $ for the product in $ \mathbb{F}^\mathcal{B} $ given by $ f g \in \mathbb{F}^\mathcal{B} $, where
\begin{equation}
(fg)(a) = (FG)(a),
\label{eq def product}
\end{equation}
for $ a \in \mathcal{B} $, $ f,g \in \mathbb{F}^\mathcal{B} $ and $ F,G \in \mathbb{F}[x; \sigma, \delta]_n $ such that $ f = E_\mathcal{B}(F) $ and $ g = E_\mathcal{B}(G) $.
\item
There exists a P-closed subset $ \Phi \subseteq \Omega $ such that $ \mathcal{W} = E_\mathcal{B}(I(\Phi)) $.
\end{enumerate}
\end{proposition}

In particular, by Item 2, skew support spaces are also vector subspaces of $ \mathbb{F}^\mathcal{B} $. Notice also that, in general, $ (fg)(a) \neq f(a) g(a) $ in Item 3 (see \cite[Th. 2.7]{lam-leroy}). 

In conclusion, in this appendix we have introduced skew supports and support spaces, and we have given the precise connections with sum-rank supports and support spaces. Except for the Hamming-metric case, the dictionary between both types of concepts depends on the choice of conjugacy representatives and P-basis of the ambient P-closed set via (\ref{eq charact of P-basis}). With this dictionary, all of the remaining results and definitions in this paper can be translated to skew supports and support spaces. We leave however as open problem defining skew supports and support spaces independently of a set of conjugacy representatives and a P-basis.

\end{document}